\documentclass[runningheads]{llncs}
\usepackage[T1]{fontenc}

\usepackage{algorithm}
\usepackage{algorithmic}
\usepackage{amsmath}
\usepackage{amsfonts}
\usepackage{amssymb}
\usepackage{color}
\usepackage{comment}
\usepackage{enumitem}
\usepackage{graphicx}
\usepackage[bookmarks=false,colorlinks=true]{hyperref}
\hypersetup{citecolor=green!50!black,linkcolor=red!70!black}
\usepackage[utf8]{inputenc}
\usepackage{listings}
\usepackage{makeidx}
\usepackage{subcaption}
\usepackage{url}
\usepackage{wrapfig}
\usepackage{xcolor}
\usepackage[pdf]{graphviz}
\usepackage{multirow}
\usepackage{refcount}
\usepackage{pbox}
\usepackage{xspace}
\usepackage{bbm}
\usepackage{booktabs}
\usepackage{stmaryrd}
\usepackage[misc]{ifsym}

\usepackage{breakcites}

\usepackage[graphicx]{realboxes}
\usepackage{tikz}
\tikzset{elliptic state/.style={draw,ellipse}}
\usetikzlibrary{shapes,shapes.geometric,arrows,fit,calc,positioning,automata,}
\usetikzlibrary{automata,arrows,decorations.pathmorphing,decorations.markings,positioning,shapes,shapes.geometric,arrows.meta,bending,fit,calc}
\usetikzlibrary{math}
\usepackage{tikz-qtree}
\usepackage{float}
\usepackage{floatflt}
\usepackage{pgfplots}
\pgfplotsset{compat=1.18}

\tikzset{
  >=stealth,
  box state/.style={draw,rectangle,minimum size=8mm},
  prob state/.style={draw,very thick,shape=circle,black,minimum size=2mm,inner sep=0mm},
  node distance=2cm,on grid,auto, initial text=,
  every loop/.style={shorten >=0pt},
  accepting/.style={double distance=1.2pt, outer sep = 0.6pt+\pgflinewidth},
  accepting dot/.style={above=-2.5pt,circle,fill,darkgreen,inner sep=2pt,radius=1pt},
  loop above/.append style={every loop/.append style={out=120, in=60, looseness=6}},
  loop below/.append style={every loop/.append style={out=300, in=240, looseness=6}},
  loop left/.append style={every loop/.append style={out=210, in=150, looseness=6}},
  loop right/.append style={every loop/.append style={out=30, in=330, looseness=6}},
  accepting arc/.style={dashed},
  marked/.style={
    dashed,
    opacity=0.3
  },
  marked on/.style={alt=#1{marked}{}},
}

\newcommand{\code}[1]{{\texttt {#1}}}

\newcommand{\R}{\mathbb{R}}
\newcommand{\E}{\mathbb{E}}
\newcommand{\Z}{\mathbb{Z}}
\newcommand{\X}{\mathbf{X}}
\newcommand{\U}{\mathbf{U}}
\newcommand{\A}{\mathcal{A}}
\newcommand{\D}{\mathcal{D}}

\newcommand{\F}{\mathbf{F}}
\newcommand{\G}{\mathbf{G}}

\newcommand{\M}{\mathcal{M}}
\newcommand{\C}{\mathcal{C}}
\newcommand{\J}{\mathcal{J}}
\newcommand{\Rc}{\mathcal{R}}

\renewcommand{\P}{\mathcal{P}}

\newcommand{\semantics}[1]{{\llbracket #1 \rrbracket}}

\newcommand{\vp}{\varphi}

\newcommand{\opt}{\text{opt}}

\newcommand{\traj}{\mathcal{Z}}

\usepackage{graphicx}
\usepackage{color}

\def\orcidID#1{{\href{http://orcid.org/#1}{\protect\raisebox{-1.25pt}{\protect\includegraphics{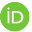}}}}}

\def\mailto#1{\textsuperscript{({\href{mailto:#1}{\color{black} \protect\raisebox{-0.8pt}{\Letter}}})}}

\begin{document}
\title{Policy Synthesis and Reinforcement Learning \\ for Discounted LTL\thanks{
This research was partially supported by ONR award N00014-20-1-2115, NSF grant CCF-2009022, and NSF CAREER award CCF-2146563.
}}
\author{
Rajeev Alur\inst{1}\orcidID{0000-0003-1733-7083} \and
Osbert Bastani\inst{1}\orcidID{0000-0001-9990-7566} \and
Kishor Jothimurugan\inst{1}\orcidID{0000-0003-1448-2947} \and
\\
Mateo Perez\inst{2}\mailto{mateo.perez@colorado.edu}\orcidID{0000-0003-4220-3212} \and
Fabio Somenzi\inst{2}\orcidID{0000-0002-2085-2003} \and
Ashutosh Trivedi\inst{2}\orcidID{0000-0001-9346-0126}
}
\authorrunning{Alur et al.}
\institute{University of Pennsylvania, Pennsylvania, USA \\
\email{\{alur,kishor,obastani\}@seas.upenn.edu} 
\and
University of Colorado Boulder, Boulder, USA \\
\email{\{mateo.perez,fabio,ashutosh.trivedi\}@colorado.edu}}
\maketitle

\begin{abstract}
    The difficulty of manually specifying reward functions has led to an interest in using linear temporal logic (LTL) to express objectives for reinforcement learning (RL). However, LTL has the downside that it is sensitive to small perturbations in the transition probabilities, which prevents probably approximately correct (PAC) learning without additional assumptions. Time discounting provides a way of removing this sensitivity, while retaining the high expressivity of the logic. We study the use of discounted LTL for policy synthesis in Markov decision processes with unknown transition probabilities, and show how to reduce discounted LTL to discounted-sum reward via a reward machine when all discount factors are identical. 
\end{abstract}

\section{Introduction}

Reinforcement learning~\cite{Sutton18} (RL) is a sampling-based approach to synthesis in systems with unknown dynamics where an agent seeks to maximize its accumulated reward. This reward is typically a real-valued feedback that the agent receives on the quality of its behavior at each step. However, designing a reward function that captures the user's intent can be tedious and error prone, and misspecified rewards can lead to undesired behavior, called \emph{reward hacking}~\cite{amodei2016concrete}.

Due to the aforementioned difficulty, recent research~\cite{fu2014probably, sadigh2014learning, hasanbeig2019reinforcement, bozkurt2020control, li2017reinforcement} has shown interest in utilizing high-level logical specifications, particularly linear temporal logic~\cite{baier2008principles} (LTL), to express intent. 
However, a significant challenge arises due to the sensitivity of LTL, similar to other infinite-horizon objectives like average reward and safety, to small changes in transition probabilities. Even slight modifications in transition probabilities can lead to significant impacts on the value, such as enabling previously unreachable states to become reachable. 
Without additional information on the transition probabilities, such as the minimum nonzero transition probability, LTL is proven to be not probably approximately correct (PAC)~\cite{kakade2003sample} learnable~\cite{alur2022framework,yang2021reinforcement}. 
Ideally, it is desirable to maintain PAC learnability while still keeping the benefits of a highly expressive temporal logic.

Discounting can serve as a solution to this problem. Typically, discounting is used to encode time-sensitive rewards (i.e., a payoff is worth more today than tomorrow), but it has a useful secondary effect that payoffs received in the distant future have small impact on the accumulated reward today. This insensitivity enables PAC learning without requiring any prior knowledge of the transition probabilities. In RL, discounted reward is commonly used and has numerous associated PAC learning algorithms~\cite{kakade2003sample}.

\begin{figure}[t]
    \centering
    \begin{tikzpicture}
        \node[state, initial above, accepting, align=center] (s0) {$s_0$};
        \node[prob state] (prob1) [right=1.5cm of s0] {};
        \node[prob state] (prob2) [left=1.5cm of s0] {};
        \node[state, align=center] (s1) [right=3cm of s0] {$s_1$};
        \node[state, align=center] (s2) [left=3cm of s0] {$s_2$};
        \path[->]
        (prob1) edge [bend right=70, swap] node {$1-p_1$} (s0)
        (prob1) edge [] node {$p_1$} (s1)
        (prob2) edge [bend left=70] node {$1-p_2$} (s0)
        (prob2) edge [swap] node {$p_2$} (s2)
        (s1) edge [loop right] node {$a_1, a_2$} ()
        (s2) edge [loop left] node {$a_1, a_2$} ()
        ;
        \path[-]
        (s0) edge [swap] node {$a_1$} (prob1)
        (s0) edge [] node {$a_2$} (prob2)
        ;
    \end{tikzpicture}
    \caption{Example showing non-robustness of safety specifications.}
    \label{fig:robust}
\end{figure}
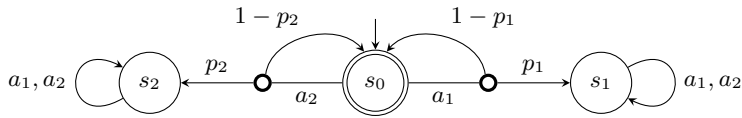

In this work, we examine the discounted LTL of \cite{almagor2014discounting} for policy synthesis in Markov decision processes (MDPs) with unknown transition probabilities.
We refer to such MDPs as ``unknown MDPs'' throughout the paper. 
This logic maintains the syntax of LTL, but discounts the temporal operators.
Discounted LTL gives a quantitative preference to traces that satisfy the objective sooner, and those that delay failure as long as possible. The authors of \cite{almagor2014discounting} examined discounted LTL in the model checking setting. Exploring policy synthesis and learnability for discounted LTL specifications  is novel to this paper.

To illustrate how discounting affects learnability, consider the example~\cite{littman2017environment} MDP shown in Figure~\ref{fig:robust}. It consists of a safe state $s_0$, two sink states $s_1,s_2$, and two actions $a_1,a_2$. Taking action $a_i$ in $s_0$ leads to a sink state with probability $p_i$ and stays in $s_0$ with probability $1-p_i$. Suppose we are interested in learning a policy to make sure that the system always stays in the state $s_0$. Now consider two scenarios---one in which $p_1 = 0$ and $p_2 = \delta$ and another in which $p_2=0$ and $p_1=\delta$ where $\delta>0$ is a small positive value. In the former case, the optimal policy is to always choose $a_1$ in $s_0$ and in the latter case, we need to choose $a_2$ in $s_0$. Furthermore, it can be shown that a near-optimal policy in one case is not near-optimal in another. 
However, we cannot select a finite number of samples needed to distinguish between the two cases (with high probability) without knowledge of $\delta$.
In contrast, the time-discounted semantics of the safety property evaluates to $1-\lambda^k$ where $k$ is the number of time steps spent in the state $s_0$. Then, for sufficiently small $\delta$, any policy achieves a high value w.r.t. the discounted safety property in both scenarios. In general, small changes to the transition probabilities do not have drastic effects on the nature of near-optimal policies for discounted interpretations of LTL properties.

\paragraph{Contributions.}
\begin{table*}
    \centering
    \begin{tabular}{lccc}
    \toprule
    \multirow{2}{*}{Specification} &
    \multirow{2}{*}{\quad Memory\quad} &
    \multicolumn{2}{c}{Policy Synthesis Algorithm} \\
    \cmidrule(lr){3-4}
    {} & {} & \quad Known MDP\quad & \quad PAC Learning\quad\\
    \midrule
    Reward Machines & Finite \cite{puterman2014markov,icarte2018using} & Exists \cite{puterman2014markov,icarte2018using} & Exists \cite{strehl2006pac}\\
    LTL & Finite \cite{baier2008principles} & Exists \cite{baier2008principles} & Impossible \cite{alur2022framework,yang2021reinforcement}\\
    Discounted LTL & Infinite  & Open & Exists  \\
    Uniformly Discounted LTL\quad\quad & Finite  & Exists  & Exists  \\
    \bottomrule\\
    \end{tabular}
    \caption{Policy synthesis in MDPs for different classes of specifications.}
    \label{tab:results}
\end{table*}
Table~\ref{tab:results} summarizes results of this paper in the context of known results regarding policy synthesis for various classes of specifications. We consider three key properties of specifications, namely, (1) whether there is a finite-state optimal policy and whether there are known algorithms for (2) computing an optimal policy when the MDP is known, as well as for (3) learning a near-optimal policy when the transition probabilities are unknown (without additional assumptions). The classes of specifications include reward machines with discounted-sum rewards \cite{icarte2018using}, linear temporal logic (LTL) \cite{baier2008principles}, discounted LTL and a variant of discounted LTL in which all discount factors are identical, which we call \emph{uniformly} discounted LTL. In this paper, we show the following.
\begin{itemize}
    \item In general, finite-memory optimal policies may not exist for discounted LTL specifications.
    \item There exists a PAC learning algorithm to learn policies for discounted LTL specifications.
    \item There is a reward machine for any uniformly discounted LTL specification such that the discounted-sum rewards capture the semantics of the specification. From this we infer that for any given MDP finite-memory optimal policies exist and can be computed.
\end{itemize}

\paragraph{Related Work.}
Linear temporal logic (LTL) is a popular and expressive formalism to unambiguously express qualitative safety and progress requirements of Kripke structures and MDPs~\cite{baier2008principles}.
The standard approach to model check LTL formulas against MDPs is the \emph{automata-theoretic} approach where the LTL formulas are first translated to a class of good-for-MDP automata~\cite{hahn2020good}, such as limit-deterministic B\"uchi automata~\cite{vardi1985automatic,sickert2016limit,sickert2016mochiba,hahn2013lazy}, and then, efficient graph-theoretic techniques (computing accepting end-component and then maximizing the probability to reach states in such components) ~\cite{de1998formal,vardi1985automatic,kwiatkowska2009prism}
over the product of the automaton with the MDP can be used to compute optimal satisfaction probabilities and strategies. 
Since LTL formulas can be translated into (deterministic) automata in doubly exponential time, the probabilistic model checking problem is in 2EXPTIME with a matching lower bound~\cite{Cour95}.

Several variants of LTL have been proposed that provide discounted temporal modalities. 
De Alfaro {et al.}~\cite{de2003discounting} proposed an extension of $\mu$-calculus with discounting and showed~\cite{de2004model} the decidability of model-checking over finite MDPs.
Mandrali~\cite{mandrali2012weighted} introduced discounting in LTL by taking a discounted sum interpretation of logic over a trace. 
Littman et al.~\cite{littman2017environment} proposed {geometric LTL} as a logic to express learning objectives in RL. However, this logic has unclear semantics for nesting operators. 
Discounted LTL was proposed by Almagor, Boker, and Kupferman~\cite{almagor2014discounting}, which considers discounting without accumulation.
The decidability of the policy synthesis problem for discounted LTL against MDPs is an open problem.

An alternative approach to discounting that ensuring PAC learnability is to introduce a fixed time horizon, along with a temporal logic for finite traces. In this setting, the logic $\text{LTL}_f$ is the most popular~\cite{camacho2019ltl,de2019foundations}. Using $\text{LTL}_f$ with a finite horizon yields simple algorithms~\cite{wells2020ltlf}, finite automata suffice for checking properties, but at the expense of the expressivity of the logic, formulas like $\G \F p$ and $\F \G p$ both mean that $p$ occurs at the end of the trace. 

There has been a lot of recent work on reinforcement learning from temporal specifications~\cite{aksaray2016q,brafman2018ltlf,de2019foundations,hasanbeig2018logically,littman2017environment,hasanbeig2019,yuan2019modular,moritz2019,ijcai2019-0557,jiang2020temporallogicbased,li2017reinforcement,icarte2018using,jothimurugan2022specification, jothimurugan2021compositional, spectrl}. Such approaches often lack strong convergence guarantees. Some methods have been developed to reduce LTL properties to discounted-sum rewards~\cite{bozkurt2020control,moritz2019} while preserving optimal policies; however they rely on the knowledge of certain parameters that depend on the transition probabilities of the unknown MDP. Recent work~\cite{alur2022framework,yang2021reinforcement,littman2017environment} has shown that PAC algorithms that do not depend on the transition probabilities do not exist for the class of LTL specifications. There has also been work on learning algorithms for LTL specifications that provide guarantees when additional information about the MDP (e.g., the smallest nonzero transition probability) is available~\cite{daca2017faster,ashok2019pac,fu2014probably}.

\section{Problem Definition}\label{sec:defn}
An alphabet $\Sigma$ is a finite set of letters. 
A finite word (resp.\ $\omega$-word)  over $\Sigma$ is defined as a finite sequence (resp.\ $\omega$-sequence) of letters from  $\Sigma$. 
We write $\Sigma^*$ and $\Sigma^\omega$ for the set of finite and $\omega$-words over $\Sigma$. 

A \emph{probability distribution} over a finite set $S$ is a function
$d \colon S {\to} [0, 1]$ such that $\sum_{s \in S} d(s) = 1$.  Let
$\D(S)$ denote the set of all discrete distributions over $S$. 

\paragraph{Markov Decision Processes.} A {Markov Decision Process} (MDP) is a tuple $\M = (S, A, s_0, P)$, where $S$ is a finite set of states,
$s_0$ is the initial state,
$A$ is a finite set of actions, and 
$P:S \times A \to \D(S)$ is the transition probability function.
An \emph{infinite run} $\psi \in (S\times A)^{\omega}$ is a sequence $\psi = s_0{a_0}s_1{a_1}\ldots$, where $s_i \in S$ and $a_i\in A$ for all $i\in\Z_{\geq 0}$. For any run $\psi$ and any $i\leq j$, we let
$\psi_{i:j}$ denote the subsequence $s_i{a_i}s_{i+1}{a_{i+1}}\ldots{a_{j-1}}s_j$. Similarly, a \emph{finite run} $h\in(S\times A)^*\times S$ is a finite sequence $h = s_0{a_0}s_1{a_1}\ldots a_{t-1}s_t$.  We use $\traj(S,A) = (S\times A)^{\omega}$ and $\traj_f(S,A) = (S\times A)^*\times S$ to denote the set of infinite and finite runs, respectively. 

A policy $\pi:\traj_f(S,A)\to\D(A)$ maps a finite run $h\in\traj_f(S,A)$ to a distribution $\pi(h)$ over actions. We denote by $\Pi(S,A)$ the set of all such policies.
A policy $\pi$ is deterministic if, for all finite runs $h\in\traj_f(S,A)$, there is an action $a\in A$ with $\pi(h)(a) = 1$.

Given a finite run $h=s_0a_0\ldots a_{t-1}s_t$, the \emph{cylinder} of $h$, denoted by $\code{Cyl}(h)$, is the set of all infinite runs with prefix $h$. Given an MDP $\M$ and a policy $\pi\in\Pi(S,A)$, we define the probability of the cylinder set by $\D^{\M}_{\pi}(\code{Cyl}(h)) = \prod_{i=0}^{t-1}\pi(h_{0:i})(a_i)P(s_i, a_i, s_{i+1})$. It is known that $\D_{\pi}^{\M}$ can be uniquely extended to a probability measure over the $\sigma$-algebra generated by all cylinder sets. Let $\P$ be a finite set of atomic propositions and $\Sigma=2^{\P}$ denote the set of all valuations of propositions in $\P$. An infinite word $\rho\in\Sigma^\omega$ is a map $\rho:\Z_{\geq 0}\to\Sigma$.

\begin{definition}[Discounted LTL]
Given a set of atomic propositions $\P$, \emph{discounted LTL} formulas over $\P$ are given by the grammar
$$\vp := b \in \P \mid \lnot\vp \mid \vp\lor\vp \mid \X_{\lambda}\vp \mid \vp~\U_{\lambda}\vp$$
where $\lambda\in[0,1)$. Note that, in general, different temporal operators within the same formula may have different discount factors $\lambda$. For a formula $\vp$ and a word $\rho=\sigma_0\sigma_1\ldots \in (2^{\P})^\omega$, the semantics $\semantics{\vp,\rho}\in[0,1]$ is given by
\begin{align*}
    \semantics{b,\rho} &= \mathbbm{1}\big(b\in\sigma_0\big)\\
    \semantics{\lnot\vp,\rho} &= 1-\semantics{\vp,\rho}\\
    \semantics{\vp_1\lor\vp_2,\rho} &= \max\big\{\semantics{\vp_1,\rho},\semantics{\vp_2,\rho}\big\}\\
    \semantics{\X_{\lambda}\vp,\rho} &= \lambda\cdot\semantics{\vp,\rho_{1:\infty}}\\
    \semantics{\vp_1\U_{\lambda}\vp_2,\rho} &= \underset{i \ge 0}{\sup} \left\{ \min\Big\{\lambda^i [\![\vp_2, \rho_{i:\infty}]\!], \underset{0\le j< i}{\min}\{\lambda^j [\![\vp_1, \rho_{j:\infty}]\!]\}\Big\}\right\}
\end{align*}
where $\rho_{i:\infty}=\sigma_i\sigma_{i+1}\ldots$ denotes the infinite word starting at position $i$. 
\end{definition}
Conjunction is defined using $\vp_1\land\vp_2 = \lnot(\lnot\vp_1\lor\lnot\vp_2)$. We use $\F_\lambda\vp = \code{true}\U_\lambda\vp$ and $\G_\lambda\vp = \lnot\F_\lambda\lnot\vp$ to denote the discounted versions of \emph{finally} and \emph{globally} operators respectively. Note that when all discount factors equal $1$, the semantics corresponds to the usual semantics of LTL. 

For this paper, we consider the case of strict discounting, where $\lambda < 1$. We refer to the case where the discount factor is the same for all temporal operators as \emph{uniform discounting}.  Our definition differs from \cite{almagor2014discounting} in two ways: 1) we discount the next operator, and 2) we enforce strict, exponential discounting.

\paragraph{Example Discounted LTL Specifications.} To develop an intuition of the semantics of discounted LTL, we now present a few example formulas and their meaning.
\begin{itemize}
    \item $\F_\lambda\, p$ 
    obtains a value of $\lambda^n$ where $n$ is the first index where $p$ becomes true in a trace, and 0 if $p$ is never true. An optimal policy attempts to reach a $p$-state as soon as possible.
    \item $\G_\lambda\, p$ obtains a value of $1-\lambda^n$ where $n$ is the first index that a $\neg p$ occurs in a trace, and 1 if $p$ always holds. 
    An optimal policy attempts to delay reaching a $\neg p$-state as long as possible.
    \item $\X_\lambda\, p$ obtains a value of $\lambda$ if $p$ is in the second position and $0$ otherwise.
    \item $p \vee \X_\lambda\, q$ obtains a value of $1$ if $p$ is in the first position of the trace, a value of $\lambda$ if the trace begins with $\neg p$ followed by $q$, and a value of $0$ otherwise.
    \item $\F_\lambda\, p \,\wedge\, \G_\lambda\, q$ evaluates to the minimum of $\lambda^n$ and $(1-\lambda^m)$, where $n$ is the first position where $p$ becomes true in a trace and $m$ is the first position where $q$ becomes false. If $n^*= {\mathit log}_\lambda 0.5$ is the index where these two competing objectives coincide, then the optimal policy attempts to stay within $q$-states for the first $n^*$ steps and then attempts to reach a $p$-state as soon as possible.  
    \item Consider the formula $\F_{\lambda_1} \G_{\lambda_2} p$. Given a trace, consider a $p$-block of length $m$ starting at position $n$, that is, $p$ holds at all positions from $n$ to $n+m-1$, and does not hold at position $n-1$ (or $n$ is the initial position). The value of such a block is  $\lambda_1^n(1-\lambda_2^m)$. The value of the trace is then the maximum over values of all such $p$-blocks. The optimal policy attempts to have as long a $p$-block as possible as early as possible. The discount factor $\lambda_1$ indicates the preference for the $p$-block to occur sooner and the discount factor $\lambda_2$ indicates the preference for the $p$-block to be longer.
    \item $\G_{\lambda_1} \F_{\lambda_2} p$ obtains a value equivalent to $\neg \F_{\lambda_1} \G_{\lambda_2} \neg p$. Traces which contain more $p$'s at shorter intervals are preferred. The discount factor $\lambda_1$ indicates the preference for the total number of $p$'s to be larger and $\lambda_2$ indicates the preference for the interval between the consecutive $p$'s to be shorter.
    
\end{itemize}

\paragraph{Policy Synthesis Problem.} Given an MDP $\M = (S,A,s_0,P)$, we assume that we have access to a \emph{labelling function} $L:S\to\Sigma$ that maps each state to the set of propositions that hold true in that state. Given any run $\psi = s_0a_0s_1a_1\ldots$ we can define an infinite word $L(\psi) = L(s_0)L(s_1)\ldots$ that denotes the corresponding sequence of labels. Given a policy $\pi$ for $\M$, we define the value of $\pi$ with respect to a discounted LTL formula $\vp$ as
\begin{equation}
    \J^{\M}(\pi, \vp)  = \displaystyle\underset{\rho\sim\D^{\M}_{\pi}}{\E}\semantics{\vp, \rho}
\end{equation}
and the optimal value for $\M$ with respect to $\vp$ as $\J^*(\M,\vp) = \sup_{\pi}\J^{\M}(\pi,\vp)$. We say that a policy $\pi$ is optimal for $\vp$ if 
$\J^{\M}(\pi,\vp) = \J^*(\M,\vp)$.
Let $\Pi_{\opt}(\M,\vp)$ denote the set of optimal policies. Given an MDP $\M$, a labelling function $L$ and a discounted LTL formula $\vp$, the policy synthesis problem is to compute an optimal policy $\pi\in\Pi_{\opt}(\M,\vp)$ when one exists.

\paragraph{Reinforcement Learning Problem.} In reinforcement learning, the transition probabilities $P$ are unknown. Therefore, we need to interact with the environment to learn a policy for a given specification. In this case, it is sufficient to learn an $\varepsilon$-optimal policy $\pi$ that satisfies 
$\J^{\M}(\pi,\vp) \geq \J^*(\M,\vp)-\varepsilon$.
We use $\Pi_{\opt}^{\varepsilon}(\M,\vp)$ to denote the set of $\varepsilon$-optimal policies. Formally, a learning algorithm $\A$ is an iterative process which, in every iteration $n$, (i) takes a step in $\M$ from the current state, (ii) outputs a policy $\pi_n$ and (iii) optionally resets the current state to $s_0$. We are interested in probably-approximately correct (PAC) learning algorithms.

\begin{definition}[PAC-MDP]
\label{def:pac-mdp}
A learning algorithm $\A$ is said to be PAC-MDP for a class of specifications $\C$ if, there is a function $\eta$ such that for any $p>0$, $\varepsilon>0$, MDP $\M=(S,A,s_0,P)$, labelling function $L$, and specification $\vp\in\C$, taking $N=\eta(|S|,|A|,|\vp|,\frac{1}{p}, \frac{1}{\varepsilon})$, with probability at least $1-p$, we have
$$\Big|\Big\{n\mid \pi_n\notin\Pi_{\opt}^{\varepsilon}(\M,\vp) \Big\}\Big|\leq N.$$
\end{definition}

It has been shown that there does not exist PAC-MDP algorithms for LTL specifications. 
Therefore, we are interested in the class of discounted LTL specifications that are strictly discounted, i.e. $\lambda < 1$ for every temporal operator.

\section{Properties of Discounted LTL}\label{sec:properties}
In this section, we discuss important properties of discounted LTL regarding the nature of optimal policies. We first show that, under uniform discounting, the amount of memory required for the optimal policy may increase with the discount factor. We then show that, in general, allowing multiple discount factors may result in optimal policies requiring infinite memory. This motivates our restriction to the uniform discounting case in Section~\ref{sec:dltl_to_rm}. We end this section by introducing a PAC learning algorithm for discounted LTL.

\subsection{Nature of Optimal Policies}
It is known that for any (undiscounted) LTL formula $\vp$ and any MDP $\M$, there exists a \emph{finite memory} policy that is optimal---i.e., the policy stores only a finite amount of information about the history. Formally, given an MDP $\M=(S,A,s_0,P)$, a finite memory policy $\pi = (M,\delta_M,\mu,m_0)$ consists of a finite set of memory states $M$, a transition function $\delta_M:M\times S\times A\to M$ and an action function $\mu:M\times S\to \D(A)$. Given a finite run $h = s_0a_0\ldots s_t = h's_t$, the policy's action is sampled from $\mu(\delta_M(m_0,h'), s_t)$ where $\delta_M$ is also used to represent the transition function extended to sequences of state-action pairs. We use $\Pi_f(S,A)$ to denote the set of finite memory policies. In this paper, we will show that uniformly discounted LTL admits finite memory optimal policies, but that infinite memory may be required for the general case.

Unlike (undiscounted) LTL, discounted LTL allows a notion of satisfaction quality. In discounted LTL, traces which satisfy a reachability objective sooner are given a higher value, and are thus preferred. If an LTL formula cannot be satisfied, the corresponding discounted LTL formula will assign higher values to traces which delay failure as long as possible. These properties of discounted LTL are desirable for enabling notions of promptness, but may yield more complex strategies which try to balance the values of multiple competing subformulas.

\begin{example}
\label{example:GandF}
Consider the discounted LTL formula $\varphi = \G_\lambda p \wedge \F_\lambda \neg p$. This formula contains two competing objectives that cannot both be completely satisfied. Increasing the value of $\G_\lambda p$ by increasing the number of $p$'s at the beginning of the trace before the first $\neg p$ decreases the value of $\F_\lambda \neg p$. Under the semantics of conjunction, the value of $\varphi$ is the minimum of the two subformulas. Specifically, the value of $\varphi$ w.r.t. a word $\rho$ is
\begin{align*}
    [\![ \G_{\lambda} p \wedge \F_{\lambda} \neg p,\rho ]\!] &= [\![ \neg \F_{\lambda} \neg p \wedge \F_{\lambda} \neg p, \rho ]\!] \\
    &= [\![ \neg( \F_{\lambda} \neg p \vee \neg \F_{\lambda} \neg p ),\rho ]\!] \\
    &= 1 - \max\{ [\![\F_{\lambda} \neg p ,\rho]\!] , [\![ \neg  \F_{\lambda} \neg p ,\rho]\!] \} \\
    &= 1 - \max\left\{ \sup_{i \ge 0} \{ \lambda^i  [\![ \neg p,\rho_{i:\infty} ]\!] \} , 1 - \sup_{i \ge 0} \{ \lambda^i  [\![ \neg p,\rho_{i:\infty} ]\!] \} \right\}.
\end{align*}
where $\rho_{i:\infty}$ is the trace starting from index $i$. Now consider a two state (deterministic) MDP with two states $S =\{s_1,s_2\}$ and two actions $A=\{a_1,a_2\}$ in which the agent can decide to either stay in $s_1$ or move to $s_2$ at any step and the system stays in $s_2$ upon reaching $s_2$. 
This MDP can be seen in Figure~\ref{fig:GandF}.
We have one proposition $p$ which holds in state $s_1$ and not in $s_2$. Note that all runs produced by the example MDP are either of the form $s_1^\omega$ or $s_1^k s_2^\omega$. The discounted LTL value of runs of the form $s_1^\omega$ is $0$. The value of runs of the form $\psi = s_1^k s_2^\omega$ is
\[
v(k) = \semantics{\vp,L(\psi)} = 1 - \max\{ \lambda^k, 1 - \lambda^k \}\enspace.
\]
A finite memory policy stays in $s_1$ for $k$ steps will yield this value. Since $\lambda^k$ is decreasing in $k$ and $1-\lambda^k$ is increasing in $k$, the integer value of $k$ that maximizes $v(k)$ lies in the interval $[\gamma-1,\gamma+1]$ where $\gamma\in\R$ satisfies $\lambda^\gamma = 1-\lambda^\gamma$. Figure~\ref{fig:GandF} shows this graphically. We have that $\gamma = \frac{\log(0.5)}{\log(\lambda)}$ which is increasing in $\lambda$. Hence, the amount of memory required increases with increase in $\lambda$. 
\end{example}

\begin{figure}[t]
\qquad
    \begin{minipage}{0.22\textwidth}
        \centering
        \begin{tikzpicture}
            \node[state, initial above, align=center] (s0) {$s_0$ \\ $p$};
            \node[state, align=center] (s1) [below=2cm of s0] {$s_1$ \\ $\bot$};
            \path[->]
            (s0) edge [loop left] node {$a_1$} ()
            (s0) edge [] node {$a_2$} (s1)
            (s1) edge [loop left] node {$a_1, a_2$} ()
            ;
        \end{tikzpicture}
    \end{minipage}%
    \qquad\qquad
    \begin{minipage}{0.30\textwidth}
        \centering
        \begin{tikzpicture}
            \pgfmathsetmacro{\lam}{0.99}
            \begin{axis}[
            domain=0:400,
            ymin=0,
            ymax=1,
            grid,
            xlabel={$k$},
            ylabel={Value},
            samples=200,
            width=6cm,
            height=4.6cm,
            legend style={at={(0.95,0.5)},anchor=east},
            extra x ticks={ln(0.5)/ln(\lam)},
            extra x tick labels={$\gamma$}
            ]
            \addplot[line width=1.4] {((x <= ln(0.5)/ln(\lam)) * (1 - \lam^x)) + ((x > ln(0.5)/ln(\lam)) * (\lam^x))};
            \addplot[blue] {1 - \lam^x};
            \addplot[red] {\lam^x};
            \addplot[dashed] coordinates {(ln(0.5)/ln(\lam), 0) (ln(0.5)/ln(\lam), 1)};
            \addlegendentry{$\G_\lambda p \wedge \F_\lambda \neg p$}
            \addlegendentry{$\G_\lambda p$}
            \addlegendentry{$\F_\lambda \neg p$}
            \end{axis}
        \end{tikzpicture}
    \end{minipage}
    \caption{
    \label{fig:GandF} 
    An example showing that memory requirements for optimal policies may depend on the discount factor. The red line is $\lambda^k$, the blue line is $1 - \lambda^k$ and the solid black line is $v(k) = 1 - \max\{ \lambda^n, 1 - \lambda^n \}$, where $k$ is the number of time steps one remains in $s_0$. The dashed vertical line shows the value $\gamma$ where $v(k)$ is maximized.
    We have set $\lambda = 0.99$. Note that changing the value of $\lambda$ corresponds to rescaling the x-axis.
    }
\end{figure}
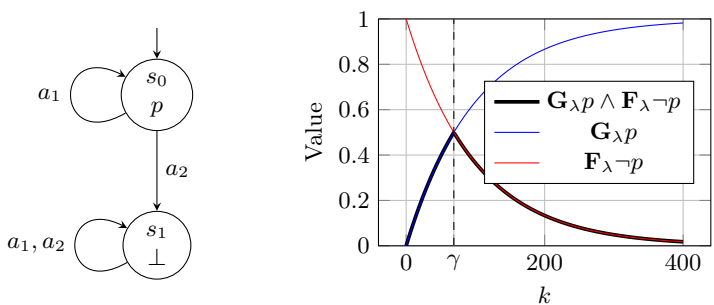

The optimal strategy in the example above tries to balance the value of two competing subformula. We will now show that extending this idea to the general case of multiple discount factors requires balancing quantities that are decaying at different speeds. This balancing may require remembering an arbitrarily long history of the trace---infinite memory is required.

\begin{theorem}\label{thm:infmem}
There exists an MDP $\M = (S,A,s_0,P)$, a labelling function $L$ and a discounted LTL formula $\vp$ such that for all $\pi\in\Pi_f(S,A)$ we have 
$
J^{\M}(\pi,\vp) < \J^*(\M,\vp).
$
\end{theorem}
\begin{proof}
Consider the MDP $\M$ depicted in Figure~\ref{fig:infmem}. It consists of three states $S=\{s_0,s_1,s_2\}$ and two actions $A=\{a_1,a_2\}$. The edges are labelled with actions and the corresponding transition probabilities. There are two propositions $\P=\{p_1,p_2\}$ and $p_1$ holds true in state $s_1$ and $p_2$ holds true in state $s_2$. The specification is given by $\vp = \F_{\lambda_1}\G_{\lambda_2}p_1\land\F_{\lambda_2}p_2$ where $\lambda_1 < \lambda_2 < 1$.

\begin{figure}[t]
    \centering
    \begin{tikzpicture}
        \node[state, initial left, align=center] (s0) {$s_0$ \\ $\bot$};
        \node[prob state] (prob) [right=1.5cm of s0] {};
        \node[state, align=center] (s1) [right=3cm of s0] {$s_1$ \\ $p_1$};
        \node[state, align=center] (s2) [right=3cm of s1] {$s_2$ \\ $p_2$};
        \path[->]
        (prob) edge [bend right=70, swap] node {$1-p$} (s0)
        (prob) edge [] node {$p$} (s1)
        (s1) edge [loop above] node {$a_1$} ()
        (s1) edge [] node {$a_2$} (s2)
        (s2) edge [loop above] node {$a_1, a_2$} ()
        ;
        \path[-]
        (s0) edge [swap] node {$a_1, a_2$} (prob)
        ;
    \end{tikzpicture}
    \caption{The need for infinite memory for achieving optimality in discounted LTL.}
    \label{fig:infmem}
\end{figure}
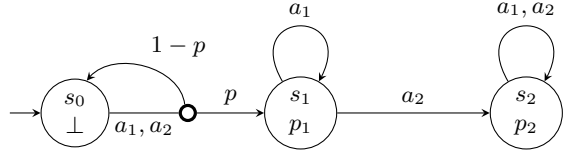

For any run $\psi$ that never visits $s_2$, we have $\semantics{\vp, L(\psi)} = 0$ since $\semantics{\F_{\lambda_2}p_2,L(\psi)} = 0$. Otherwise the run has the form $\psi=s_0^{k_0} s_1^{k_1} s_2^{\omega}$ where $k_0$ is stochastic and $k_1$ is a strategic choice by the agent. To show that this requires an infinite amount of memory to play optimally, one just has to show that the optimal choice of $k_1$ increases with $k_0$. This means that the agent must remember $k_0$, the number of steps spent in the initial state, via an unbounded counter. Note that every value of $k_0$ has a non-zero probability in $\M$ and therefore choosing a suboptimal $k_1$ for even a single value of $k_0$ causes a decrease in value from the policy that always chooses optimal $k_1$.

The value of the run $\psi$ is $\semantics{\vp,L(\psi)}=\min( \lambda_1^{k_0} (1 - \lambda_2^{k_1}) , \lambda_2^{k_0+k_1} )$. Note that $\lambda_1^{k_0} (1 - \lambda_2^{k_1})$ increases with increase in $k_1$ and $\lambda_2^{k_0+k_1}$ decreases with increase in $k_1$. Therefore taking $\gamma\in\R$ to be such that $\lambda_1^{k_0} (1 - \lambda_2^{\gamma}) = \lambda_2^{k_0+\gamma}$, the optimal choice of $k_1$ lies in the interval $[\gamma-1,\gamma+1]$. Now $\gamma$ satisfies $1 = \big((\lambda_2 / \lambda_1)^{k_0} + 1\big)\lambda_2^{\gamma}$. Since $\lambda_1<\lambda_2<1$ we must have that $\gamma$ increases with increase in $k_0$. Therefore, $k_1$ also increases with increase in $k_0$. \hfill\qed
\end{proof}

\subsection{PAC Learning}
In the above discussion, we showed that one might need infinite memory to act optimally w.r.t a discounted LTL formula. However, it can be shown that for any MDP $\M$, labelling function $L$, discounted LTL formula $\vp$ and any $\varepsilon>0$, there is a finite-memory policy $\pi$ that is $\varepsilon$-optimal for $\vp$.
In fact, we can show that this class of discounted LTL formulas admit a PAC-MDP learning algorithm.

\begin{theorem}[Existence of PAC-MDP]
There exists a PAC-MDP learning algorithm for discounted LTL specifications.
\end{theorem}

\begin{proof}[sketch]
Our approach to compute $\varepsilon$-optimal policies for discounted LTL is to compute a policy which is optimal for $T$ steps. The policy will depend on the entire history of atomic propositions that has occured so far.

Given discounted LTL specification $\vp$, the first step of the algorithm is to determine $T$. We select $T$ such that for any two infinite words $\alpha$ and $\beta$ where the first $T+1$ indices match, i.e. $\alpha_{0:T} = \beta_{0:T}$, we have that $\big|[\![ \varphi,\alpha ]\!] - [\![ \varphi,\beta ]\!]\big| \le \varepsilon$. Say that the maximum discount factor appearing in all temporal operators is $\lambda_{\max{}}$. Due to the strict discounting of discounted LTL, selecting $T \geq \frac{\log \varepsilon}{\log \lambda_{\max}}$ ensures that $\big|[\![ \varphi,\alpha ]\!] - [\![ \varphi,\beta ]\!]\big| \le \lambda^n \le \varepsilon$.

Now we unroll the MDP for $T$ steps. We include the history of the atomic proposition sequence in the state. Given an MDP $\M = (S, A, s_0, P)$ and a labeling $L: S \to \Sigma$, the unrolled MDP $\M_T = (S', A', s_0', P')$ is such that 
\[
S' = \bigcup_{t = 0}^T S \times \underbrace{\Sigma \times \ldots \times \Sigma}_\text{$t$ times}\enspace,
\]

\noindent $A' = A$, 
$
P'((s, \sigma_0, \ldots, \sigma_{t-1}), a, (s', \sigma_0, \ldots, \sigma_{t-1}, \sigma_t)) = P(s, a, s')
$
if $0 \le t \le T$ and  $\sigma_t = L(s')$, and is $0$ otherwise (the MDP goes to a sink state if $t>T$). The leaves of the unrolled MDP are the states where $T$ timesteps have elapsed. In these states, there is an associated finite word of length $T$. For a finite word of length $T$, we define the value of any formula $\varphi$ to be zero beyond the end of the trace, i.e. $[\![ \varphi,\rho_{j:\infty} ]\!] = 0$ for any $j > T$. We then compute the value of the finite words associated with the leaves which is then considered as the reward at the final step. We can use existing PAC algorithms to compute an $\varepsilon$-optimal policy w.r.t. this reward for the finite horizon MDP $\M_T$ from which we can obtain a $2\varepsilon$-optimal policy for $\M$ w.r.t the specification $\vp$.\hfill\qed

\end{proof}

\section{Uniformly Discounted LTL to Reward Machines}\label{sec:dltl_to_rm}

In general, optimal strategies for discounted LTL require infinite memory (Theorem~\ref{thm:infmem}). However, producing such an example required the use of multiple, varied discount factors. In this section, we will show that finite memory is sufficient for optimal policies under uniform discounting, where the discount factors for all temporal operators in the formula are the same. We will also provide an algorithm for computing these strategies.

Our approach is to reduce uniformly discounted LTL formulas to \emph{reward machines}, which are finite state machines in which each transition is associated with a reward. We show that the value of a given discounted LTL formula $\vp$ for an infinite word $\rho$ is the discounted-sum reward computed by a corresponding reward machine.

Formally, a reward machine is a tuple $\Rc = (Q, \delta, r, q_0, \lambda)$ where $Q$ is a finite set of states, $\delta: Q\times \Sigma\to Q$ is the transition function, $r:Q\times\Sigma\to\R$ is the reward function, $q_0\in Q$ is the initial state, and $\lambda\in[0,1)$ is the discount factor. With any infinite word $\rho=\sigma_0\sigma_1\ldots\in\Sigma^\omega$, we can associate a sequence of rewards $c_0c_1\ldots$ where $c_t = r(q_t,\sigma_t)$ with $q_t = \delta(q_{t-1}, \sigma_{t-1})$ for $t>0$. We use $\Rc(\rho)$ to denote the discounted reward achieved by $\rho$,
$$\Rc(\rho) = \sum_{t=0}^\infty\lambda^tc_t,$$
and $\Rc(w)$ to denotes the partial discounted reward achieved by the finite word $w=\sigma_0\sigma_1\ldots\sigma_T\in\Sigma^*$---i.e., $\Rc(w) = \sum_{t=0}^T\lambda^tc_t$ where $c_t$ is the reward at time $t$. 

Given a reward machine $\Rc$ and an MDP $\M$, our objective is to maximize the expected value $\Rc(\rho)$ from the reward machine reading the word $\rho$ produced by the MDP. Specifically, the value for a policy $\pi$ for $\M$ is 
\[
\J^{\M}(\pi, \Rc)  = \displaystyle\underset{\rho\sim\D^{\M}_{\pi}}{\E}[\Rc(\rho)]
\]
where $\pi$ is optimal if $\J^{\M}(\pi, \Rc) = \sup_{\pi} \J^{\M}(\pi, \Rc)$. Finding such an optimal policy is straightforward: we consider the product of the reward machine $\Rc$ with the MDP $\M$ to form a product MDP with a discounted reward objective. In the corresponding product MDP, we can compute optimal policies for maximizing the expected discounted-sum reward using standard techniques such as policy iteration and linear programming. If the transition function of the MDP is unknown, this product can be formed on-the-fly and any RL algorithm for discounted reward can be applied. Using the state space of the reward machine as memory, we can then obtain a finite-memory policy that is optimal for $\Rc$. 

We have the following theorem showing that we can construct a reward machine $\Rc_{\vp}$ for every uniformly discounted LTL formula $\vp$.

\begin{theorem}\label{thm:rm_exists}
    For any uniformly discounted LTL formula $\vp$, in which all temporal operators use a common discount factor $\lambda$, we can construct a reward machine $\Rc_\varphi = (Q, \delta, r, q_0, \lambda)$ such that for any $\rho\in\Sigma^\omega$, we have $\Rc_\varphi(\rho) = \semantics{\rho, \varphi}$.
\end{theorem}

We provide the reward machine construction for Theorem~\ref{thm:rm_exists} in the next subsection. Using this theorem, one can use a reward machine $\Rc_\vp$ that matches the value of a particular uniformly discounted LTL formula $\vp$, and then apply the procedure outlined above for computing optimal finite-memory policies for reward machines.

\begin{corollary}
For any MDP $\M$, labelling function $L$ and a discounted LTL formula $\vp$ in which all temporal operators use a common discount factor $\lambda$, there exists a finite-memory optimal policy $\pi\in\Pi_{\opt}(\M,\vp)$. Furthermore, there is an algorithm to compute such a policy.
\end{corollary}

\subsection{Reward Machine Construction}
For our construction, we examine the case of uniformly discounted LTL formula with positive discount factors $\lambda \in (0, 1)$. This allows us to divide by $\lambda$ in our construction. We note that the case of uniformly discounted LTL formula with $\lambda = 0$ can be evaluated after reading the initial letter of the word, and thus have trivial reward machines.

The reward machine $\Rc_\vp$ constructed for the uniformly discounted LTL formula $\vp$ exhibits a special structure. Specifically, all edges within any given strongly-connected component (SCC) of $\Rc_{\vp}$ share the same reward, which is either $0$ or $1-\lambda$, while all other rewards fall within the range of $[0,1-\lambda]$.
We present an inductive construction of the reward machines over the syntax of discounted LTL that maintains these invariants.

\begin{lemma}\label{lem:rm_const}
    For any uniformly discounted LTL formula $\varphi$ there exists a reward machine $\Rc_\varphi = (Q, \delta, r, q_0, \lambda)$ such that following hold:
    \begin{enumerate}
        \item [$I_1$.] For any $\rho\in\Sigma^\omega$, we have $\Rc_\varphi(\rho) = \semantics{\rho, \varphi}$.
        \item [$I_2$.] There is a partition of the states $Q = \bigcup_{\ell=1}^L Q_\ell$ and a type mapping $\chi:[L]\to \{0, 1-\lambda\}$ such that for any $q\in Q_\ell$ and $\sigma\in\Sigma$, 
        \begin{enumerate}
            \item $\delta(q,\sigma)\in\bigcup_{m=\ell}^L Q_m$, and
            \item if $\delta(q,\sigma)\in Q_\ell$ then $r(q, \sigma) = \chi(\ell)$.
        \end{enumerate}       
        \item [$I_3$.] For any $q\in Q$ and $\sigma\in\Sigma$, we have $0\leq r(q,\sigma)\leq 1-\lambda$.
    \end{enumerate}
\end{lemma}

Our construction proceeds inductively. We define the reward machine for the base case of a single atomic proposition, i.e. $\vp = p$, and then the construction for negation, the next operator, disjunction, the eventually operator (for ease of presentation), and the until operator. The ideas used in the constructions for disjunction, the eventually operator, and the until operator build off of each other, as they all involve keeping track of the maximum/minimum value over a set of subformulas. We use properties $I_1$ and $I_3$ to show correctness, and properties $I_2$ and $I_3$ to show finiteness. A summary of the construction is presented in Appendix~\ref{ssec:summary}. \\

\noindent\textbf{Atomic propositions.}
Let ${\varphi = p}$ for some $p\in\P$. The reward machine $\Rc_\varphi = (Q, \delta, r, q_0, \lambda)$ for $\varphi$ is such that $Q=\{q_0, q_1, q_2\}$ and $\delta(q,\sigma) = q$ for all $q\in \{q_1,q_2\}$ and $\sigma\in\Sigma$. The reward machine is shown in Figure~\ref{fig:atomic_rm} where edges are labelled with propositions and rewards. If $p\in\sigma$, $\delta(q_0, \sigma) = q_1$ and $r(q_0, \sigma) = 1-\lambda$. If $p\notin\sigma$, $\delta(q_0, \sigma) = q_2$ and $r(q_0, \sigma) = 0$. Finally, $r(q_1, \sigma) = 1-\lambda$ and $r(q_2, \sigma) = 0$ for all $\sigma\in\Sigma$. It is clear to see that $I_1$, $I_2$, and $I_3$ hold.\\

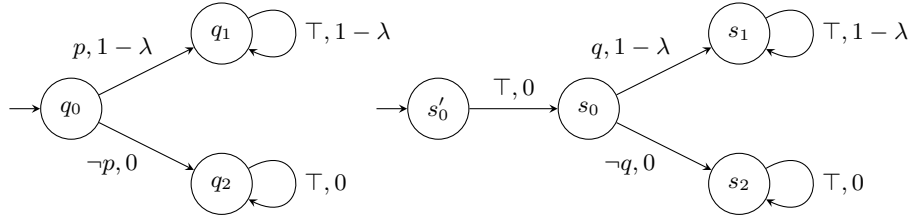
\begin{figure}[t]
    \begin{minipage}{0.4\textwidth}
        \centering
        \begin{tikzpicture}
            \node[state, initial left] (q0) {$q_0$};
            \node[state] (q1) [above right=1cm and 2cm of q0] {$q_1$};
            \node[state] (q2) [below right=1cm and 2cm of q0] {$q_2$};
            \path[->]
            (q0) edge [] node [yshift=0.5mm, xshift=2mm] {$p, 1-\lambda$} (q1)
            (q0) edge [swap] node {$\neg p, 0$} (q2)
            (q1) edge [loop right] node {$\top, 1-\lambda$} ()
            (q2) edge [loop right] node {$\top, 0$} ()
            ;
        \end{tikzpicture}
    \end{minipage}%
    \begin{minipage}{0.6\textwidth}
        \centering
        \begin{tikzpicture}
            \node[state, initial left] (q0prime) {$s_0'$};
            \node[state] (q0) [right=2cm of q0prime] {$s_0$};
            \node[state] (q1) [above right=1cm and 2cm of q0] {$s_1$};
            \node[state] (q2) [below right=1cm and 2cm of q0] {$s_2$};
            \path[->]
            (q0prime) edge [] node [] {$\top, 0$} (q0)
            (q0) edge [] node [yshift=0.5mm, xshift=2mm] {$q, 1-\lambda$} (q1)
            (q0) edge [swap] node {$\neg q, 0$} (q2)
            (q1) edge [loop right] node {$\top, 1-\lambda$} ()
            (q2) edge [loop right] node {$\top, 0$} ()
            ;
        \end{tikzpicture}
    \end{minipage}
    \caption{Reward machines for $\vp = p$ (left) and $\vp = \X_\lambda q$ (right). The transitions are labeled by the guard and reward.}
    \label{fig:atomic_rm}
\end{figure}

\noindent\textbf{Negation.}
Let $\varphi = \lnot \varphi_1$ for some LTL formula $\varphi_1$ and let  $\Rc_{\varphi_1} = (Q, \delta, r, q_0, \lambda)$ be the reward machine for $\varphi_1$. 
Notice that the reward machine for $\varphi$ can be constructed from $\Rc_{\varphi_1}$ by simply replacing every reward $c$ with $(1-\lambda) - c$ as $\sum_{i=0}^\infty\lambda^i(1-\lambda) = 1$.  Formally, $\Rc_{\varphi} = (Q, \delta, r', q_0, \lambda)$ where $r'(q, \sigma) = (1-\lambda) - r(q, \sigma)$ for all $q\in Q$ and $\sigma\in \Sigma$. Again, assuming that invariants $I_1$, $I_2$, and $I_3$ hold for $\Rc_{\varphi_1}$, it easily follows that they hold for $\Rc_{\varphi}$.\\

\noindent\textbf{Next operator.}
Let $\varphi = \mathbf{X}_{\lambda}\varphi_1$ for some $\varphi_1$ and let $\Rc_{\varphi_1} = (Q, \delta, r, q_0, \lambda)$ be the reward machine for $\varphi_1$. The reward machine for $\varphi$ can be constructed from $\Rc_{\varphi_1}$ by adding a new initial state $q_0'$ and a transition in the first step from it to the initial state of $\Rc_{\vp_1}$. 
From the next step $\Rc_{\varphi}$ simulates $\Rc_{\varphi_1}$.
This has the resulting effect of skipping the first letter, and decreasing the value by $\lambda$.
Formally, $\Rc_{\varphi} = (\{q_0'\}\sqcup Q, \delta', r', q_0', \lambda)$ where $\delta'(q_0', \sigma) = q_0$ and $\delta'(q, \sigma)=\delta(q,\sigma)$ for all $q\in Q$ and $\sigma\in \Sigma$. Similarly, $r'(q_0', \sigma) = 0$ and $r'(q,\sigma) = r(q, \sigma)$ for all $q\in Q$ and $\sigma\in\Sigma$. Assuming that invariants $I_1$, $I_2$, and $I_3$ hold for $\Rc_{\varphi_1}$, it follows that they hold for $\Rc_{\varphi}$.\\

\noindent\textbf{Disjunction.}
Let $\varphi = \varphi_1\lor\varphi_2$ for some $\vp_1,\vp_2$ and let $\Rc_{\varphi_1} = (Q_1, \delta_1, r_1, q_0^1, \lambda)$ and $\Rc_{\varphi_2} = (Q_2, \delta_2, r_2, q_0^2, \lambda)$ be the reward machines for $\varphi_1$ and $\varphi_2$, respectively.
The reward machine $\Rc_{\varphi} = (Q, \delta, r, q_0, \lambda)$ is constructed $\Rc_{\varphi_1}$ and $\Rc_{\varphi_2}$ such that  for any finite word it maintains the invariant that the discounted reward is the maximum of the reward provided by $\Rc_{\varphi_1}$ and $\Rc_{\varphi_2}$. Moreover, once it is ascertained that the reward provided by one machine cannot be overtaken by the other for any suffix, $\Rc_{\varphi}$ begins simulating the reward machine with higher reward. 

The construction involves a product construction along with a real-valued component that stores a scaled difference between the total accumulated reward for $\varphi_1$ and $\varphi_2$. In particular,
$Q = (Q_1\times Q_2\times \R)\sqcup Q_1\sqcup Q_2$ and $q_0=(q_0^1, q_0^2, 0)$. The reward deficit $\zeta$ of a state $q = (q_1,q_2,\zeta)$ denotes the difference between the total accumulated reward for $\varphi_1$ and $\varphi_2$ divided by $\lambda^{n}$ where $n$ is the total number of steps taken to reach $q$. The reward function is defined as follows.
\begin{itemize}
    \item  For $q = (q_1,q_2,\zeta)$, we let $f(q,\sigma) = r_1(q_1, \sigma) - r_2(q_2, \sigma) + \zeta$ denote the new (scaled) difference between the discounted-sum rewards accumulated by $\Rc_{\vp_1}$ and $\Rc_{\vp_2}$. The current reward depends on whether $f(q,\sigma)$ is positive (accumulated reward from $\Rc_{\vp_1}$ is higher) or negative and whether the sign is different from $\zeta$. Formally,
    \begin{align*}
        r(q, \sigma) = \begin{cases} r_1(q_1,\sigma) + \min\{0, \zeta\} &\text{if}\ f(q,\sigma) \geq 0\\
        r_2(q_2,\sigma) - \max\{0, \zeta\} &\text{if}\ f(q,\sigma) < 0
        \end{cases}
    \end{align*}
    \item For a state $q_i\in Q_i$ we have $r(q_i, \sigma) = r_i(q_i, \sigma)$.
\end{itemize}
Now we need to make sure that $\zeta$ is updated correctly. We also want the transition function to be such that the (reachable) state space is finite and the reward machine satisfies $I_1,I_2$ and $I_3$.
\begin{itemize}
    \item First, we make sure that, when the difference $\zeta$ is too large, the machine transitions to the appropriate state in $Q_1$ or $Q_2$. For a state $q = (q_1, q_2, \zeta)$ with $|\zeta| \geq 1$, we have
    \begin{align*}
        \delta(q, \sigma) = \begin{cases}
            \delta_1(q_1, \sigma) &\text{if}\ \zeta \geq 1\\
            \delta_2(q_2, \sigma) &\text{if}\ \zeta \leq -1.
        \end{cases}
    \end{align*}
    \item For states with $|\zeta|< 1$, we simply advance both the states and update $\zeta$ accordingly.
    Letting $f(q,\sigma) = r_1(q_1, \sigma)-r_2(q_2, \sigma)+\zeta$, we have that for a state $q = (q_1, q_2, \zeta)$ with $|\zeta| < 1$,
    \begin{align}\label{eq:or_trans}
        \delta(q, \sigma) &= (\delta_1(q_1, \sigma), \delta_2(q_2,\sigma), f(q,\sigma)/\lambda).
    \end{align}
    \item Finally, for $q_i\in Q_i$, $\delta(q_i,\sigma) = \delta_i(q_i, \sigma)$.
\end{itemize}

\paragraph{Finiteness.} We argue that the (reachable) state space of $\Rc_{\vp}$ is finite. Let $Q_i = \bigcup_{\ell=1}^{L_i} Q_\ell^i$ for $i\in\{1,2\}$ be the SCC decompositions of $Q_1$ and $Q_2$ that satisfy property $I_2$ for $\Rc_{\vp_1}$ and $\Rc_{\vp_2}$ respectively. Intuitively, if $\Rc_{\vp}$ stays within $Q_{\ell}^1\times Q_{m}^2\times\R$ for some $\ell\leq L_1$ and $m\leq L_2$, then the rewards from $\Rc_{\vp_1}$ and $\Rc_{\vp_2}$ are constant; this enables us to infer the reward machine ($\Rc_{\vp_1}$ and $\Rc_{\vp_2}$) with the higher total accumulated reward in a finite amount of time after which we transition to $Q_1$ or $Q_2$. Hence the set of all possible values of $\zeta$ in a reachable state $(q_1,q_2,\zeta)\in Q_{\ell}^1\times Q_{m}^2\times\R$ is finite. Formally, we show this by induction on $\ell+m$ in Lemma~\ref{lem:dis_fin}.

\paragraph{Property $I_1$.} Intuitively, it suffices to show that $\Rc_{\vp}(w) = \max\{\Rc_{\vp_1}(w),\Rc_{\vp_2}(w)\}$ for every finite word $w\in\Sigma^*$. We show this property along with the fact that for any $w\in\Sigma^*$ of length $n$, if the reward machine reaches a state $(q_1,q_2,\zeta)$, then $\zeta = (\Rc_{\varphi_1}(w) - \Rc_{\varphi_2}(w))/\lambda^{n}$. This can be proved using induction on $n$ and the full proof can be found in Appendix~\ref{sec:dis_const}.

\paragraph{Property $I_2$.} This property is true if and only if for every SCC $\C$ of $\Rc_{\varphi}$ there is a type $c\in\{0,1-\lambda\}$ such that if $\delta(q,\sigma) = q'$ for some $q,q'\in\C$ and $\sigma\in\Sigma$, we have $r(q,\sigma) = c$. From the definition of the transition function $\delta$, $\C$ cannot contain two states where one is of the form $(q_1,q_2,\zeta)\in Q_1\times Q_2\times\R$ and the other is $q_i\in Q_i$ for some $i\in\{1,2\}$. Now if $\C$ is completely contained in $Q_i$ for some $i$, we can conclude from the inductive hypothesis that the rewards within $\C$ are constant (and they are all either $0$ or $1-\lambda$). When all states of $\C$ are contained in $Q_1\times Q_2\times \R$, they must be contained in $\bar{Q}_1\times\bar{Q}_2\times\R$ where $\bar{Q}_i$ is some SCC of $\Rc_{\vp_i}$. In such a case, we can show that $|\C|=1$ and in the presence of a self loop on a state within $\C$, the reward must be either $0$ or $1-\lambda$. More details in Appendix~\ref{sec:dis_const}.

\paragraph{Property $I_3$.} We now show that all rewards are bounded between $0$ and $(1-\lambda)$. Let $q = (q_1,q_2,\zeta)$ and $f(q,\sigma) = r_1(q_1,\sigma) - r_2(q_2,\sigma) + \zeta$. We show the bound for the case when $f(q,\sigma)\geq 0$ and the other case is similar. If $\zeta\geq 0$, then $r(q,\sigma) = r_1(q_1,\sigma) \in [0,1-\lambda]$. If $\zeta < 0$, then $r(q,\sigma)\leq r_1(q_1,\sigma)\leq 1-\lambda$ and
\[           
r(q,\sigma) = r_1(q_1,\sigma) + \zeta = f(q,\sigma) + r_2(q_2,\sigma) \geq 0.
\]
This concludes the construction for $\varphi_1\lor\varphi_2$. \\

\noindent\textbf{Eventually operator.} For ease of presentation, we treat the until operator as a generalization of the \emph{eventually} operator $\F_{\lambda}$ and present it first. We have that $\vp = \F \vp_1$ for some $\vp_1$. Let $\Rc_{\varphi_1} = (Q_1, \delta_1, r_1, q_0^1, \lambda)$ be the reward machine for $\vp_1$. Let $\X_\lambda^i$ denote the operator $\X_\lambda$ applied $i$ times. We begin by noting that
\[
    \F_\lambda \vp_1 \equiv \bigvee_{i\ge 0} \X_\lambda^i \vp_1 = \vp_1 \vee \X_\lambda \vp_1 \vee \X_\lambda^2 \vp_1 \vee \ldots.
\]
The idea of the construction is to keep track of the unrolling of this formula up to the current timestep $n$,
\begin{align*}
    \F_\lambda^n \vp_1 &= \bigvee_{n \ge i\ge 0} \X_\lambda^i \vp_1 = \vp_1 \vee \X_\lambda \vp_1 \vee \X_\lambda^2 \vp_1 \vee \ldots \vee \X_\lambda^n \vp_1.
\end{align*}
For this, we will generalize the construction for disjunction. In the disjunction construction, there were states of the form $(q_1, q_2, \zeta)$ where $\zeta$ was a bookkeeping parameter that kept track of the difference between $\Rc_{\vp_1}(w)$ and $\Rc_{\vp_2}(w)$, namely, $\zeta = (\Rc_{\vp_1}(w) - \Rc_{\vp_2}(w))/\lambda^n$  where $w \in \Sigma^*$ is some finite word of length $n$. To generalize this notion to make a reward machine for $\max\{\Rc_{1}, \ldots, \Rc_{k}\}$, we will have states of the form $\{(q_1, \zeta_1), \ldots, (q_n, \zeta_n)\}$ where $\zeta_i = (\Rc_{i}(w) - \max_{j}\Rc_{j}(w))/\lambda^n$. When $\zeta_i \le -1$ then $\Rc_{i}(w) + \lambda^n \le \max_{j}\Rc_{j}(w)$ and we know that the associated reward machine $\Rc_{i}$ cannot be the maximum, so we drop it from our set. 
We also note that the value of $\X_\lambda^i \vp_1$ can be determined by simply waiting $i$ steps before starting the reward machine $\Rc_{\varphi_1}$, i.e. $\lambda^i \Rc_{\vp_1} (\rho_{i:\infty}) = \Rc_{\X_\lambda^i \vp_1} (\rho)$. This allows us to perform a subset construction for this operator. 

For a finite word $w=\sigma_0\sigma_1\ldots\sigma_n \in \Sigma^*$ and a nonnegative integer $k$, let $w_{k:\infty}$ denote the subword $\sigma_k\ldots\sigma_n$ which equals the empty word $\epsilon$ if $k {>} n$. We use the notation $\semantics{\X_\lambda^k \vp_1, w} = \lambda^k \Rc_{\vp_1}(w_{k:\infty})$ and define $\semantics{\F_\lambda^k \vp_1, w} = \max_{k \ge i \ge 0} \semantics{\X_\lambda^k \vp_1, w}$ which represents the maximum value accumulated by the reward machine of some formula of the form $\X_\lambda^i\vp_1$ with $i\leq k$ on a finite word $w$. The reward machine for $\F_\lambda\vp_1$ will consist of states of the form $(v, S)$, containing a value $v$ for bookkeeping and a set $S$ that keeps track of the states of all $\Rc_{\X_\lambda^i\vp_1}$ that may still obtain the maximum given a finite prefix $w$ of length $n$, i.e. reward machine states of all subformulas $\X_\lambda^i \vp_1$ for $n \ge i \ge 0$ that satisfy $\semantics{\X_\lambda^i \vp_1, w} + \lambda^n > \semantics{\F_\lambda^n \vp_1, w}$ since $\lambda^n$ is the maximum additional reward obtainable by any $\rho\in\Sigma^\omega$ with prefix $w$. The subset $S$ consists of elements of the form $(q_i, \zeta_i) \in S$ where $q_i = \delta_1(q_0^1, w_{i:\infty})$ and $\zeta_i = (\semantics{\X_\lambda^i \vp_1, w} - \semantics{\F_\lambda^n \vp_1, w})/\lambda^n$ corresponding to each subformula $\X_\lambda^i \vp_1$.
The value $v = \max\{-1, -\semantics{\F_\lambda^n \vp_1, w}/\lambda^n\}$ is a bookkeeping parameter used to initialize new elements in the set $S$ and to stop adding elements to $S$ when $v \le -1$. 
We now present the construction formally.

We form a reward machine $\Rc_\vp = (Q, \delta, r, q_0, \lambda)$ where $Q = \mathbb{R} \times 2^{Q_1 \times \mathbb{R}}$ and $q_0 = (0, \{(q_0^1, 0)\})$.       
We define a few functions that ease defining our transition function. Let $f(\zeta, q, \sigma) = r_1(q, \sigma) + \zeta$ and $m(S, \sigma) = \max\limits_{(q_i, \zeta_i) \in S} f(\zeta_i, q_i, \sigma)$. For the subset construction, we define
\[
\Delta(S, \sigma) = \bigcup_{(q, \zeta) \in S} \{(\delta_1(q, \sigma), \zeta') : \zeta' = \big((f(\zeta, q, \sigma) - m(S,\sigma))/\lambda \big)> -1\}
\]
The transition function is
\[
\delta((v, S), \sigma) = 
    \begin{cases}
        \Big(v'(S,v,\sigma),\ \Delta(S, \sigma) \sqcup \big(q_0^1, v'(S,v,\sigma)\big)\Big) & \text{if\ } v'(S,v,\sigma) > -1 \\
        (-1,\ \Delta(S, \sigma)) & \text{if\ } v'(S,v,\sigma) \le -1
    \end{cases}
\]
where $v'(S,v,\sigma) = (v - m(S,\sigma))/\lambda$. The reward function is $r((v, S), \sigma) = m(S, \sigma)$.

We now argue that $\Rc_{\vp}$ satisfies properties $I_1$, $I_2$ and $I_3$ and the set of reachable states in $\Rc_{\vp}$ is finite assuming $\Rc_{\vp_1}$ satisfies $I_1$, $I_2$ and $I_3$. Details are in Appendix~\ref{sec:ev_const}.

\paragraph{Finiteness.}
Consider states of the form $(v, S) \in Q$. If $v = 0$, then it must be that $\zeta_i = 0$ for all $(q_i, \zeta_i) \in S$ since receiving a non-zero reward causes the value of $v$ to become negative. There are only finitely many such states. If $-1 < v < 0$, then we will reach a state $(v', S') \in Q$ with $v' = -1$ in at most $n$ steps, where $n$ is such that $v/\lambda^n \le -1$. Therefore, the number of reachable states $(v,S)$ with $-1<v<0$ is also finite. Also, the number of states of the form $(-1,S)$ that can be initially reached (via paths consisting only of states of the form $(v,S')$ with $v>-1$) is finite. Furthermore, upon reaching such a state $(-1,S)$, the reward machine is similar to that of a disjunction (maximum) of $|S|$ reward machines. From this we can conclude that the full reachable state space is finite. See Appendix~\ref{sec:ev_const} for a full proof.

\paragraph{Property $I_1$.}
The transition function is designed so that the following holds true: for any finite word $w \in \Sigma^*$ of length $n$ and letter $\sigma \in \Sigma$, if $\delta(q_0, w) = (v, S)$, then $m(S, \sigma) = (\semantics{\F_\lambda^{n+1} \vp_1, w\sigma} - \semantics{\F_\lambda^n \vp_1, w})/\lambda^n$. Since $r((v, S), \sigma) = m(S, \sigma)$, we get that $\Rc_\vp(w) = \semantics{\F_\lambda^n \vp_1, w}$. Thus, $\Rc_\vp(\rho) = \semantics{\F_\lambda \vp_1, \rho}$ for any infinite word $\rho \in \Sigma^\omega$. This property for $m(S, \sigma)$ follows from the preservation of all the properties outlined in the above description of the construction.

\paragraph{Property $I_2$.}
Consider an SCC $\C$ in $\Rc_\vp$ such that $(v, S) = \delta((v,S), w)$ for some $(v, S) \in \C$ and $w \in \Sigma^*$ of length $n > 0$. Note that if $-1 < v < 0$, then $(v', S') = \delta((v, S), w)$ is such that $v' < v$. Thus, it must be that $v = 0$ or $v = -1$. If $v = 0$, then all the reward must be zero, since any nonzero rewards result in $v < 0$. If $v = -1$, then it must be that for any $(q_i, \zeta_i) \in S$, $q_i$ is in an SCC $\C_1^i$ in $\Rc_{\vp_1}$ with some reward type $c_i \in \{0, 1-\lambda\}$. For all $\zeta_i$ to remain fixed (which is necessary as otherwise some $\zeta_i$ strictly increases or decreases), it must be that all $c_i$ are the same, say $c$. Thus, the reward type in $\Rc_{\vp_1}$ for SCC $\C$ equals $c$.

\paragraph{Property $I_3$.}
We can show that for any finite word $w \in \Sigma^*$ of length $n$ and letter $\sigma \in \Sigma$, if $\delta(q_0, w) = (v, S)$, then the reward is $r((v, S), \sigma) = m(S, \sigma) = (\semantics{\F_\lambda^{n+1} \vp_1, w\sigma} - \semantics{\F_\lambda^n \vp_1, w})/\lambda^n$ using induction on $n$. Since property $I_3$ holds for $\Rc_{\vp_1}$, we have that $0 \le (\semantics{\F_\lambda^{n+1} \vp_1, w\sigma} - \semantics{\F_\lambda^n \vp_1, w}) \le (1-\lambda)\lambda^n$.\\

\noindent\textbf{Until operator.} We now present the until operator, generalizing the ideas presented for the eventually operator. We have that $\vp = \vp_1 \U_\lambda \vp_2$ for some $\vp_1$ and $\vp_2$. Let $\Rc_{\vp_1} = (Q_1, \delta_1, r_1, q_0^1, \lambda)$ and $\Rc_{\vp_2} = (Q_2, \delta_2, r_2, q_0^2, \lambda)$. 
Note that
\begin{align*}
    \vp_1 \U_\lambda \vp_2 &= \bigvee_{i\ge0} ( \X_\lambda^i \vp_2 \wedge \vp_1 \wedge \X_\lambda \vp_1 \wedge \ldots \wedge \X_\lambda^{i-1} \vp_1) \\
    &= \vp_2 \vee (\X_\lambda \vp_2 \wedge \vp_1) \vee (\X_\lambda^2 \vp_2 \wedge \vp_1 \wedge \X_\lambda \vp_1) \vee \ldots .
\end{align*} 
The goal of the construction is to keep track of the unrolling of this formula up to the current timestep $n$, 
\begin{align*}
    \vp_1 \U_\lambda^n \vp_2 &= \bigvee_{n\ge i\ge0} (\X_\lambda^i \vp_2 \wedge \vp_1 \wedge \X_\lambda \vp_1 \wedge \ldots \wedge \X_\lambda^{i-1} \vp_1)
    = \bigvee_{n\ge i\ge0} \psi_i .
\end{align*}
Each $\psi_i$ requires a subset construction in the style of the eventually operator construction to maintain the minimum. We then nest another subset construction in the style of the eventually operator construction to maintain the maximum over $\psi_i$. 
For a finite word $w \in \Sigma^*$, we use the notation $\semantics{\psi_i, w}$ and $\semantics{\vp_1 \U_\lambda^{k} \vp_2, w}$ for the value accumulated by reward machine corresponding to these formula on the word $w$, i.e. $\semantics{\psi_i, w} = \min\{\semantics{\X_\lambda^i \vp_2}, \min_{i > j \ge 0} \{ \semantics{\X_\lambda^j \vp_1, w} \}$ and $\semantics{\vp_1 \U_\lambda^{k} \vp_2, w} = \max_{k \ge i \ge 0} \semantics{\psi_i, w}$.

Let $\mathcal{S} = 2^{(Q_1 \sqcup Q_2) \times \mathbb{R}}$ be the set of subsets containing $(q, \zeta)$ pairs, where $q$ may be from either $Q_1$ or $Q_2$.
The reward machine consists of states of the form $(v, I, \mathcal{X})$ 
where the value $v \in \mathbb{R}$ and the subset $I \in \mathcal{S}$ are for bookkeeping, and $\mathcal{X} \in 2^{\mathcal{S}}$ is a subset of subsets for each $\psi_i$. Specifically, each element of $\mathcal{X}$ is a subset $S$ corresponding to a particular $\psi_i$ which may still obtain the maximum, i.e. $\semantics{\psi_i, w} + \lambda^n > \semantics{\vp_1 \U_\lambda^{n} \vp_2, w}$. Each element of $S$ is of the form $(q, \zeta)$. We have that $q \in Q_2$ for at most one element where $q = \delta_2(q_0^2, w_{k:\infty})$ and $\zeta = (\semantics{\X_\lambda^k \vp_2, w} - \semantics{\vp_1 \U_\lambda^{n} \vp_2, w})/\lambda^n$. For the other elements of $S$, we have that $q \in Q_1$ with $q = \delta_1(q_0^1, w_{k:\infty})$ and $\zeta = (\semantics{\X_\lambda^k \vp_1, w} - \semantics{\vp_1 \U_\lambda^{n} \vp_2, w})/\lambda^n$. If for any of these elements, the value of its corresponding formula becomes too large to be the minimum for the conjunction forming $\psi_i$, i.e. $\semantics{\psi_i, w} + \lambda^n \le \semantics{\vp_1 \U_\lambda^{n} \vp_2, w} + \lambda^n \le \semantics{\X_\lambda^k \vp_t, w}$ which occurs when $\zeta \ge 1$, that element is dropped from $S$. In order to update $\mathcal{X}$, we add a new $S$ corresponding to $\psi_n$ on the next timestep. The value $v = \max\{-1, \semantics{\vp_1 \U_\lambda^{n} \vp_2, w}\}$ is a bookkeeping parameter for initializing new elements in the subsets and for stopping the addition of new elements when $v \le -1$. The subset $I$ is a bookkeeping parameter that keeps track of the subset construction for $\bigwedge_{n > i \ge 0} \X_\lambda^i \vp_1$, which is used to initialize the addition of a subset corresponding to $\psi_n = \X_\lambda^n \vp_2 \wedge (\bigwedge_{n > i \ge 0} \X_\lambda^i \vp_1)$.
We now define the reward machine formally.

We define a few functions that ease defining our transition function. 
We define $\delta_*(q, \sigma) = \delta_i(q, \sigma)$ and $f_*(\zeta, q, \sigma)= r_i(q, \sigma) + \zeta$ if $q \in Q_i$ for $i \in \{1, 2\}$.
We also define $n(S, \sigma) = \min_{(q_i, \zeta_i) \in S} f_*(\zeta_i, q_i, \sigma)$ and $m(\mathcal{X}, \sigma) = \max_{S \in \mathcal{X}} n(S, \sigma)$. For the subset construction, we define
\[
\Delta(S, \sigma, m) = \bigcup_{(q, \zeta) \in S} \{(\delta_*(q, \sigma), \zeta') :  \zeta' < 1\}
\]
where $\zeta' = (f_*(\zeta, q, \sigma) - m)/\lambda$ and
\[
T(\mathcal{X}, \sigma, m) = \bigcup_{ S \in \mathcal{X}}\{\Delta(S, \sigma,m ) :  n(S, \sigma) > -1\} .
\]

We form a reward machine $\Rc_\vp = (Q, \delta, r, q_0, \lambda)$ where $Q = \mathbb{R} \times \mathcal{S} \times 2^\mathcal{S}$ and $q_0 = (0, \emptyset, \{\{(q_0^2, 0)\}\})$. The transition function is
\[
\delta((v, I, \mathcal{X}), \sigma) = 
    \begin{cases}
        \Big(v',\ I',\ T(\mathcal{X}, \sigma, m) \sqcup \big(I' \sqcup (q_0^2, v')\big)\Big) & \text{if\ } v' > -1 \\
        (-1,\ \emptyset,\ T(\mathcal{X}, \sigma, m)) & \text{if\ } v' \le -1
    \end{cases}
\]
where $m = m(\mathcal{X},\sigma)$, $v' = (v - m)/\lambda$, and $I' = \Delta(I \sqcup (q_0^1, v'), \sigma, m)$. The reward function is $r((v, I, \mathcal{X}), \sigma) = m(\mathcal{X}, \sigma)$. 

We now show a sketch of correctness, which mimics the proof for the eventually operator closely.

\paragraph{Finiteness.} 
Consider states of the form $(v, I, \mathcal{X}) \in Q$. If $v = 0$, then for all $S \in \mathcal{X}$ and $(q_i, \zeta_i) \in S$ it must be that $\zeta_i = 0$ since receiving a non-zero reward causes the value of $v$ to become negative. Similarly, all $\zeta_i = 0$ for $(q_i, \zeta_i) \in I$ when $v = 0$. There are only finitely many such states. If $-1 < v < 0$, then we will reach a state $(v', I', \mathcal{X}') \in Q$ with $v' = -1$ in at most $n$ steps, where $n$ is such that $v/\lambda^n \le -1$. Therefore, the number of reachable states $-1<v<0$ is also finite. Additionally, the number of states where $v = -1$ that can be initially reached 
is finite. Upon reaching such a state $(-1,\emptyset,\mathcal{X}')$,
the reward machine is similar to that of the finite disjunction of reward machines for finite conjunctions. 

\paragraph{Property $I_1$.} The transition function is designed so that the following holds true: for any finite word $w \in \Sigma^*$ of length $n$ and letter $\sigma \in \Sigma$, if $\delta(q_0, w) = (v, I, \mathcal{X})$, then $m(\mathcal{X}, \sigma) = (\semantics{\vp_1 \U_\lambda^{n+1} \vp_2, w\sigma} - \semantics{\vp_1 \U_\lambda^n \vp_2, w})/\lambda^n$. Since $r((v, I, \mathcal{X}), \sigma) = m(\mathcal{X}, \sigma)$, we get that $\Rc_\vp(w) = \semantics{\vp_1 \U_\lambda^n \vp_2, w}$. Thus, $\Rc_\vp(\rho) = \semantics{\vp_1 \U_\lambda \vp_2, \rho}$ for any infinite word $\rho \in \Sigma^\omega$. This property for $m(\mathcal{X},\sigma)$ follows from the properties outlined in the construction, which can be shown inductively. See Lemma~\ref{lem:u_const}.

\paragraph{Property $I_2$.} Consider an SCC $\C$ of $\Rc_{\vp}$ and a state $(v, I, \mathcal{X}) \in \C$. If $v = 0$, then we must receive zero reward because non-zero reward causes the value of $v$ to become negative. It cannot be that $-1 < v < 0$ since if $v < 0$, we reach a state $(v', I', \mathcal{X}') \in Q$ with $v' = -1$ in at most $n$ steps, where $n$ is such that $v/\lambda^n \le -1$. If $v = -1$, then we have a state of the form $(-1, \emptyset, \mathcal{X})$.
For this to be an SCC, all elements of the form $(q_k, \zeta_k) \in S$ for $S \in \mathcal{X}$ must be such that $q_k$ is in an SCC of its respective reward machine (either $\Rc_{\vp_1}$ or $\Rc_{\vp_2}$) with reward type $t_k \in \{0, 1-\lambda\}$. Additionally, there cannot be a $t_k' \neq t_k$ otherwise there would be a $\zeta_k$ that changes following a cycle in the SCC $\C$. Thus, the reward for this SCC $\C$ is $t_k$.

\paragraph{Property $I_3$.} This property can be shown by recalling the property above that $r((v, I, \mathcal{X}), \sigma) = m(\mathcal{X}, \sigma) = (\semantics{\vp_1 \U_\lambda^{n+1} \vp_2, w\sigma} - \semantics{\vp_1 \U_\lambda^n \vp_2, w})/\lambda^n$.

\section{Conclusion}\label{sec:conclusions}
This paper studied policy synthesis for discounted LTL in MDPs with unknown transition probabilities. Unlike LTL, discounted LTL provides an insensitivity to small perturbations of the transitions probabilities which enables PAC learning without additional assumptions. We outlined a PAC learning algorithm for discounted LTL that uses finite memory. We showed that optimal strategies for discounted LTL require infinite memory in general due to the need to balance the values of multiple competing objectives. To avoid this infinite memory, we examined the case of uniformly discounted LTL, where the discount factors for all temporal operators are identical. We showed how to translate uniformly discounted LTL formula to finite state reward machines. 
This construction shows that finite memory is sufficient, and provides an avenue to use discounted reward algorithms, such as reinforcement learning, for computing optimal policies for uniformly discounted LTL formulas.

\bibliographystyle{splncs04}
\bibliography{papers}
\newpage
\appendix
\section{Reward Machine Construction}

This appendix contains details regarding the inductive construction of reward machines from uniformly discounted LTL specifications that were omitted in the main paper. 
The construction for propositions, negation, and the $\X_\lambda$ operator are straightforward and is presented in detail in the main paper.
We begin by summarizing the construction, followed by details for disjunction, and the eventually and until operators.

\subsection{Summary of Reward Machine Construction}
\label{ssec:summary}
\noindent\textbf{Atomic propositions.} Let ${\varphi = p}$ for some $p\in\P$. The reward machine $\Rc_\varphi = (Q, \delta, r, q_0, \lambda)$ for $\varphi$ is such that
\begin{itemize}
    \itemsep1em
    \item $Q=\{q_0, q_1, q_2\}$
    \item $\delta(q,\sigma) = 
    \begin{cases}
        q_1 & \text{if\ } q = q_1 \lor (q = q_0 \land p \in \sigma) \\
        q_2 & \text{if\ } q = q_2 \lor (q = q_0 \land p \notin \sigma) \\
    \end{cases}$
    \item $r(q, \sigma) =
    \begin{cases}
        1 - \lambda & \text{if\ } q = q_1 \lor (q = q_0 \land p \in \sigma) \\
        0 & \text{if\ } q = q_2 \lor (q = q_0 \land p \notin \sigma) \\
    \end{cases}$ .
\end{itemize} 

\noindent\textbf{Negation.} Let $\varphi = \lnot \varphi_1$ for some LTL formula $\varphi_1$ and let $\Rc_{\varphi_1} = (Q, \delta, r, q_0, \lambda)$ be the reward machine for $\varphi_1$. Then, the reward machine $\Rc_{\varphi} = (Q, \delta, r', q_0, \lambda)$ for $\varphi$ is such that
\begin{itemize}
    \itemsep1em
    \item $r'(q, \sigma) = (1-\lambda) - r(q, \sigma)$ .
\end{itemize}

\noindent\textbf{Next operator.} Let $\varphi = \mathbf{X}_{\lambda}\varphi_1$ for some $\varphi_1$ and let $\Rc_{\varphi_1} = (Q, \delta, r, q_0, \lambda)$ be the reward machine for $\varphi_1$. Then, the reward machine $\Rc_{\varphi} = (Q', \delta', r', q_0', \lambda)$ for $\varphi$ is such that
\begin{itemize}
    \itemsep1em
    \item $Q' = \{q_0'\}\sqcup Q$
    \item $\delta'(q, \sigma) = 
    \begin{cases}
        q_0 & \text{if\ } q = q_0' \\
        \delta(q, \sigma) & \text{if\ } q \in Q
    \end{cases}$
    \item $r'(q, \sigma) = 
    \begin{cases}
        0 & \text{if\ } q = q_0' \\
        r(q, \sigma) & \text{if\ } q \in Q
    \end{cases}$ 
    \item $q_0'$ is a new state where $q_0' \notin Q$ .
\end{itemize}

\noindent\textbf{Disjunction.} Let $\varphi = \varphi_1\lor\varphi_2$ for some $\vp_1$ and $\vp_2$, and let $\Rc_{\varphi_1} = (Q_1, \delta_1, r_1, q_0^1, \lambda)$ and $\Rc_{\varphi_2} = (Q_2, \delta_2, r_2, q_0^2, \lambda)$ be the reward machines for $\varphi_1$ and $\varphi_2$, respectively. Then, the reward machine $\Rc_{\varphi} = (Q, \delta, r, q_0, \lambda)$ for $\varphi$ is such that
\begin{itemize}
    \itemsep1em
    \item $Q = (Q_1\times Q_2\times \R)\sqcup Q_1\sqcup Q_2$
    \item $\delta(q, \sigma) = 
    \begin{cases}
        \delta_1(q_1, \sigma) & \text{if\ } q = (q_1, q_2, \zeta) \land \zeta \geq 1 \\
        \delta_2(q_2, \sigma) & \text{if\ } q = (q_1, q_2, \zeta) \land \zeta \leq -1 \\
        (\delta_1(q_1, \sigma), \delta_2(q_2,\sigma), f(q,\sigma)/\lambda) & \text{if\ } q = (q_1, q_2, \zeta) \land \left|\zeta\right| < 1 \\
        \delta_1(q, \sigma) & \text{if\ } q \in Q_1 \\
        \delta_2(q, \sigma) & \text{if\ } q \in Q_2
    \end{cases}$
    \item $r(q, \sigma) = 
    \begin{cases}
        r_1(q_1,\sigma) + \min\{0, \zeta\} & \text{if\ } q = (q_1, q_2, \zeta) \land f(q,\sigma) \geq 0 \\
        r_2(q_2,\sigma) - \max\{0, \zeta\} & \text{if\ } q = (q_1, q_2, \zeta) \land f(q,\sigma) < 0 \\
        r_1(q, \sigma) & \text{if\ } q \in Q_1 \\
        r_2(q, \sigma) & \text{if\ } q \in Q_2
    \end{cases}$
    \item $q_0=(q_0^1, q_0^2, 0)$
\end{itemize}
where we have defined the following function to assist with the definition
\begin{itemize}
    \itemsep1em
    \item $f(q,\sigma) = r_1(q_1, \sigma) - r_2(q_2, \sigma) + \zeta$ 
    \begin{itemize}[label=]
        \item where $q = (q_1, q_2, \zeta)$ .
    \end{itemize}
\end{itemize}

\noindent\textbf{Eventually operator.} Let $\vp = \F \vp_1$ for some $\vp_1$, and let $\Rc_{\varphi_1} = (Q_1, \delta_1, r_1, q_0^1, \lambda)$ be the reward machine for $\vp_1$. Then, the reward machine $\Rc_\vp = (Q, \delta, r, q_0, \lambda)$ for $\vp$ is such that
\begin{itemize}
    \itemsep1em
    \item $Q = \mathbb{R} \times 2^{Q_1 \times \mathbb{R}}$
    \item $\delta((v, S), \sigma) = 
    \begin{cases}
        \Big(v'(S,v,\sigma),\ \Delta(S, \sigma) \sqcup \big(q_0^1, v'(S,v,\sigma)\big)\Big) & \text{if\ } v'(S,v,\sigma) > -1 \\
        (-1,\ \Delta(S, \sigma)) & \text{if\ } v'(S,v,\sigma) \le -1
    \end{cases}$
    
    \begin{itemize}[label=]
        \item where $v'(S,v,\sigma) = (v - m(S,\sigma))/\lambda$
    \end{itemize}
    \item $r((v, S), \sigma) = m(S, \sigma)$
    \item $q_0 = (0, \{(q_0^1, 0)\})$
\end{itemize}
where we have defined the following functions to assist with the definition
\begin{itemize}
    \itemsep1em
    \item $f(\zeta, q, \sigma) = r_1(q, \sigma) + \zeta$
    \item $m(S, \sigma) = \max\limits_{(q_i, \zeta_i) \in S} f(\zeta_i, q_i, \sigma)$
    \item $\Delta(S, \sigma) = \bigcup_{(q, \zeta) \in S} \{(\delta_1(q, \sigma), \zeta') : \zeta' = \big((f(\zeta, q, \sigma) - m(S,\sigma))/\lambda \big)> -1\}$ .
\end{itemize}

\noindent\textbf{Until operator.} Let $\vp = \vp_1 \U_\lambda \vp_2$ for some $\vp_1$ and $\vp_2$, and let $\Rc_{\vp_1} = (Q_1, \delta_1, r_1, q_0^1, \lambda)$ and $\Rc_{\vp_2} = (Q_2, \delta_2, r_2, q_0^2, \lambda)$ be the reward machines for $\varphi_1$ and $\varphi_2$, respectively. Then, the reward machine $\Rc_\vp = (Q, \delta, r, q_0, \lambda)$ for $\vp$ is such that
\begin{itemize}
    \itemsep1em
    \item $Q = \mathbb{R} \times \mathcal{S} \times 2^\mathcal{S}$
    \item $\delta((v, I, \mathcal{X}), \sigma) = 
    \begin{cases}
        \Big(v',\ I',\ T(\mathcal{X}, \sigma, m) \sqcup \big(I' \sqcup (q_0^2, v')\big)\Big) & \text{if\ } v' > -1 \\
        (-1,\ \emptyset,\ T(\mathcal{X}, \sigma, m)) & \text{if\ } v' \le -1
    \end{cases}$
    
    \begin{itemize}[label=]
        \item where $m = m(\mathcal{X},\sigma)$, $v' = (v - m)/\lambda$, and $I' = \Delta(I \sqcup (q_0^1, v'), \sigma, m)$
    \end{itemize}
    \item $r((v, I, \mathcal{X}), \sigma) = m(\mathcal{X}, \sigma)$
    \item $q_0 = (0, \emptyset, \{\{(q_0^2, 0)\}\})$
\end{itemize}
where we have defined the following functions to assist with the definition
\begin{itemize}
    \itemsep1em
    \item $\mathcal{S} = 2^{(Q_1 \sqcup Q_2) \times \mathbb{R}}$
    \item $\delta_*(q, \sigma) = 
    \begin{cases}
        \delta_1(q, \sigma) & \text{if\ } q \in Q_1 \\
        \delta_2(q, \sigma) & \text{if\ } q \in Q_2
    \end{cases}$
    \item $f_*(\zeta, q, \sigma) = 
    \begin{cases}
        r_1(q, \sigma) + \zeta & \text{if\ } q \in Q_1 \\
        r_2(q, \sigma) + \zeta & \text{if\ } q \in Q_2
    \end{cases}$
    \item $n(S, \sigma) = \min_{(q_i, \zeta_i) \in S} f_*(\zeta_i, q_i, \sigma)$
    \item $m(\mathcal{X}, \sigma) = \max_{S \in \mathcal{X}} n(S, \sigma)$
    \item $\Delta(S, \sigma, m) = \bigcup_{(q, \zeta) \in S} \{(\delta_*(q, \sigma), \zeta') :  \zeta' < 1\}$
    
    \begin{itemize}[label=]
        \item where $\zeta' = (f_*(\zeta, q, \sigma) - m)/\lambda$
    \end{itemize}
    \item $T(\mathcal{X}, \sigma, m) = \bigcup_{ S \in \mathcal{X}}\{\Delta(S, \sigma,m ) :  n(S, \sigma) > -1\}$ .
\end{itemize}

\subsection{Construction for Disjunction}\label{sec:dis_const}

We provide the omitted proofs for the reward machine construction for the $\lor$ operator. An example of a reward machine constructed using our algorithm is shown in Figure~\ref{fig:disjunction}. First, we show finiteness of the state space.

\subsubsection{Finiteness.}  Let $Q_i = \bigcup_{\ell=1}^{L_i} Q_\ell^i$ for $i\in\{1,2\}$ be the SCC decompositions of $Q_1$ and $Q_2$ that satisfy property $I_2$ for $\Rc_{\vp_1}$ and $\Rc_{\vp_2}$ respectively. For $\ell\in[L_1]$ and $m\in[L_2]$, define $$\Gamma(\ell, m) = \{\zeta\mid (q_1,q_2,\zeta) \ \text{is reachable}\ \text{from}\ \text{initial}\ \text{state}\  \&\  q_1\in Q_\ell^1\  \&\  q_2\in Q_{m}^2\}.$$ The following lemma shows that $\Gamma(\ell, m)$ is finite for all $\ell$ and $m$.
\begin{lemma}\label{lem:dis_fin}
For all $\ell\in[L_1]$ and $m\in[L_2]$, $|\Gamma(\ell, m)| < \infty$.
\end{lemma}
\begin{proof}
    Proof by induction on $\ell + m$. The base case follows from the argument for the general inductive case. Assume that the lemma holds for all $\ell$ and $m$ with $\ell+m\leq k$. Suppose we have $\ell$ and $m$ such that $\ell+m = k+1$. For any $q_1\in Q_\ell^1$, $q_2\in Q_m^2$ and $\zeta\in\R$, $(q_1,q_2,\zeta)$ is reachable if and only if it can be reached in a single step from $(q_1', q_2', \zeta')$ where $q_1'\in Q_{\ell'}^1$, $q_2'\in Q_{m'}^2$ and $\zeta'\in \Gamma(\ell',m')$ for some $\ell'\leq \ell$ and $m'\leq m$.
    \begin{itemize}
        \item If $\ell' < \ell$ or $m' < m$, IH ensures that the possible values of $\zeta$ is finite. Let $\Gamma_0(\ell,m)\subseteq\Gamma(\ell, m)$ denote the possible values of $\zeta$ achievable via such transitions.\\
        \item If $\ell' = \ell$ and $m'=m$,  $(q_1,q_2, \zeta)$ is reachable from some state $(q_1^0,q_2^0,\zeta_0)$ where $q_1^0\in Q_{\ell}^1$, $q_2^0\in Q_{m}^2$ and $\zeta_0\in \Gamma_0(\ell,m)$. There are two cases to consider.\\
        
        \begin{itemize}
            \item Suppose $\chi_1(\ell) \neq \chi_2(m)$ ($\chi_1$ and $\chi_2$ are the type mappings of partitions in $\Rc_{\varphi_1}$ and $\Rc_{\varphi_2}$). WLOG, let $\chi_1(\ell) = 1-\lambda$ and $\chi_2(m) = 0$. Then consider a path $(q_1^0,q_2^0,\zeta_0),\ldots,(q_1^n,q_2^n,\zeta_n)$ within $Q_{\ell}^1\times Q_m^2\times\Gamma(\ell,m)$. It has to be the case that $|\zeta_p| < 1$ for all $p < n$. Furthermore, $$\zeta_p = \sum_{t=1}^p \frac{1-\lambda}{\lambda^t} + \frac{\zeta_0}{\lambda^p} = \frac{1}{\lambda^p}\big(\sum_{t=0}^{p-1}\lambda^t(1-\lambda) + \zeta_0\big).$$ Therefore $\zeta_p\to\infty$ as $p\to\infty$ and we must have that the paths within $Q_{\ell}^1\times Q_m^2\times\Gamma(\ell,m)$ have bounded length. Hence $\Gamma(\ell, m)$ must be finite.\\
            
            \item Suppose $\chi_1(\ell) = \chi_2(m)$. Let $(q_1^0,q_2^0,\zeta_0),\ldots,(q_1^n,q_2^n,\zeta_n)$ be a path within $Q_{\ell}^1\times Q_m^2\times\Gamma(\ell,m)$. If $\zeta_0 = 0$, then $\zeta_p = 0$ for all $p\leq n$. Otherwise, $\zeta_p = \frac{\zeta_0}{\lambda^p}$ for all $p\leq n$. Since $|\zeta_p| < 1$ for all $p < n$, the length of such paths in $Q_\ell^1\times Q_\ell^2\times \Gamma(\ell, m)$ have bounded length. Therefore, $\Gamma(\ell, m)$ must be finite.
        \end{itemize}
    \end{itemize}
    Hence, we can conclude that $|\Gamma(\ell, m)| < \infty$ for all $\ell$ and $m$.\hfill\qed
\end{proof}
\begin{figure}[t]
    \centering
    \begin{tikzpicture}
        \node[state, initial left] (a0) {$a_0$};
        \node[state] (a1) [above right=1.2cm and 2cm of a0] {$a_1$};
        \node[state] (a2) [above right=0.6cm and 2.75cm of a1] {$a_3$};
        \node[state] (a3) [right=3cm of a2] {$a_7$};
        \node[state] (a4) [below right=1.2cm and 2cm of a3] {$q_1$};
        \node[state] (a5) [below right=1.2cm and 2cm of a0] {$a_2$};
        \node[state] (a8) [below right=0.6cm and 2.75cm of a1] {$a_4$};
        \node[state] (a6) [below right=0.6cm and 2.75cm of a5] {$a_6$};
        \node[state] (a7) [above right=0.6cm and 2.75cm of a5] {$a_5$};
        \node[state] (a9) [right=3cm of a7] {$a_8$};
        \node[state] (a10) [below right=1.2cm and 2cm of a9] {$s_1$};
        \path[->]
        (a0) edge [] node [] {$p, 1/3$} (a1)
        (a0) edge [swap] node [] {$\neg p, 0$} (a5)
        (a1) edge [] node [xshift=3.1mm] {$q, 1/3$} (a2)
        (a1) edge [swap] node [xshift=3.3mm] {$\neg q, 1/3$} (a8)
        (a2) edge [] node [] {$\top, 1/3$} (a3)
        (a3) edge [swap] node [yshift=1mm] {$\top, 1/3$} (a4)
        (a4) edge [loop above] node {$\top, 1/3$} ()
        (a8) edge [swap] node [anchor=north] {$\top, 1/3$} (a4)
        (a5) edge [] node [xshift=3.1mm] {$q, 1/3$} (a7)
        (a5) edge [swap] node [xshift=3.1mm] {$\neg q, 0$} (a6)
        (a6) edge [loop right] node {$\top, 0$} ()
        (a7) edge [] node [] {$\top, 1/3$} (a9)
        (a9) edge [swap] node [yshift=1mm] {$\top, 1/3$} (a10)
        (a10) edge [loop above] node {$\top, 1/3$} ()
        ;
    \end{tikzpicture}
    \caption{The reward machine for $\vp = p \vee \X_\lambda q$ via the disjunction construction with $\lambda = 2/3$. See Figure~\ref{fig:atomic_rm} for the reward machines for $p$ and $\X_\lambda q$ used in the construction. The transitions are labeled by the guard and reward. The states are
    $a_0 = (q_0, s_0', 0)$,
    $a_1 = (q_1, s_0, {1}/{2})$,
    $a_2 = (q_2, s_0, 0)$,
    $a_3 = (q_1, s_1, {3}/{4})$,
    $a_4 = (q_1, s_2, {5}/{4})$,
    $a_5 = (q_2, s_1, {-1}/{2})$,
    $a_6 = (q_2, s_2, 0)$,
    $a_7 = (q_1, s_1, {9}/{8})$, and 
    $a_8 = (q_2, s_1, {-5}/{4})$.
    }
    \label{fig:disjunction}
\end{figure}
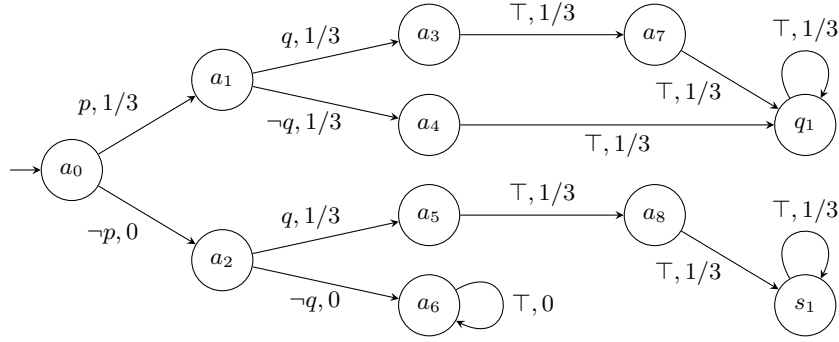
\subsubsection{Property $I_1$.} We show that the $\Rc_{\varphi}$ computes the maximum of $\Rc_{\varphi_1}$ and $\Rc_{\varphi_2}$. We begin with the following lemma.
\begin{lemma}
    For any finite word $w\in\Sigma^*$ of length $n$, if $\delta(q_0, w) = (q_1,q_2,\zeta)$ then $(i)$ $\zeta = (\Rc_{\varphi_1}(w) - \Rc_{\varphi_2}(w))/\lambda^{n}$ and $(ii)$ $\Rc_{\varphi}(w) = \max\{\Rc_{\varphi_1}(w), \Rc_{\varphi_2}(w)\}$.
\end{lemma}
\begin{proof}
    By induction on $n$.
    \begin{itemize}
        \item \textit{Base case ($n=0$)}. $\delta(q_0, w) = (q_0^1, q_0^2, 0)$ and we have $$\Rc_{\varphi}(w) = \Rc_{\varphi_1}(w) = \Rc_{\varphi_2}(w) = 0.$$
        \item \textit{Inductive case ($n>0$)}. Let $w = u\sigma$ where $\sigma\in\Sigma$. Let $\delta(q_0, u) = q' = (q_1',q_2',\zeta')$ and $f(q',\sigma) = r_1(q_1',\sigma) - r_2(q_2',\sigma) + \zeta'$. Then $\Rc_{\varphi_i}(w) = \Rc_{\varphi_i}(u) + \lambda^{n-1}r_i(q_i', \sigma)$ for $i\in\{1,2\}$. Hence, from IH we get that 
        \begin{align*}
        \Rc_{\varphi_1}(w) - \Rc_{\varphi_2}(w) &= \Rc_{\varphi_1}(u) - \Rc_{\varphi_2}(u) + \lambda^{n-1}(r_1(q_1',\sigma) - r_2(q_2',\sigma))\\
        &=\lambda^{n-1}\zeta' + \lambda^{n-1}(r_1(q_1',\sigma) - r_2(q_2',\sigma))\\
        &= \lambda^{n-1}f(q',\sigma)\\
        &= \lambda^{n}\zeta
        \end{align*}
        where the last equality followed from Equation~\ref{eq:or_trans}. We only show $(ii)$ for the case when $f(q',\sigma) \geq 0$ (the other case is similar). Note that when $f(q',\sigma)\geq 0$ we have $\Rc_{\varphi_1}(w)\geq\Rc_{\varphi_2}(w)$. Using the IH we get
        \begin{align*}
            \Rc_{\varphi}(w) &= \Rc_{\varphi}(u) + \lambda^{n-1}r(q', \sigma)\\
            &= \max\{\Rc_{\varphi_1}(u), \Rc_{\varphi_2}(u)\} + \lambda^{n-1}(r_1(q_1',\sigma) + \min\{0, \zeta'\}).
        \end{align*}
        There are two cases to consider. If $\zeta' \geq 0$, then $$\Rc_{\varphi}(w) = \Rc_{\varphi_1}(u) + \lambda^{n-1}r_1(q_1',\sigma) = \Rc_{\varphi_1}(w) = \max\{\Rc_{\varphi_1}(w), \Rc_{\varphi_2}(w)\}.$$ If $\zeta' < 0$, we have
        \begin{align*}
            \Rc_{\varphi}(w) &= \Rc_{\varphi_2}(u) + \lambda^{n-1}\Big(r_1(q_1',\sigma) + \frac{\Rc_{\varphi_1}(u) - \Rc_{\varphi_2}(u)}{\lambda^{n-1}}\Big)\\
            &= \Rc_{\varphi_1}(u) + \lambda^{n-1}r_1(q_1',\sigma)\\
            &= \Rc_{\varphi_1}(w) = \max\{\Rc_{\varphi_1}(w),\Rc_{\varphi_2}(w)\}.
        \end{align*}
    \end{itemize}
    Thus the lemma holds by induction.\hfill\qed
\end{proof}
Now whenever the automaton leaves the state space $Q_1\times Q_2\times \R$, we have that the last state in $Q_1\times Q_2\times \R$ is $q = (q_1,q_2,\zeta)$ with $|\zeta|\geq 1$. Let us consider the case when $\zeta\geq 1$ (the other case is similar). Let $w$ be a word of length $n$ with $\delta(q_0, w) = q$. From the above lemma we have $\lambda^n\zeta = \Rc_{\varphi_1}(w) - \Rc_{\varphi_2}(w)$. Therefore, for any $\rho\in\Sigma^{\omega}$, using the fact that all rewards are bounded between $0$ and $1-\lambda$, we have
\begin{align*}
    \Rc_{\varphi_1}(w\rho) - \Rc_{\varphi_2}(w\rho) &\geq \Rc_{\varphi_1}(w) - \Rc_{\varphi_2}(w) - \lambda^n\Big(\sum_{t=0}^{\infty}\lambda^t(1-\lambda)\Big)\\
    &= \lambda^n(\zeta - 1)
    \geq 0.
\end{align*}
Hence $\Rc_{\varphi}(w) = \max\{\Rc_{\varphi_1}(w),\Rc_{\varphi_2}(w)\}$ for all finite words $w$. For any $\rho\in \Sigma^{\omega}$, using $\rho_{t}$ to denote the prefix of $\rho$ of length $t$,
\begin{align*}
    \Rc_{\varphi}(\rho) &= \lim_{t\to\infty}\Rc_{\varphi}(\rho_{t})\\
    &= \lim_{t\to\infty}\max\{\Rc_{\varphi_1}(\rho_{t}),\Rc_{\varphi_2}(\rho_{t})\}\\
    &= \max\Big\{\lim_{t\to\infty}\Rc_{\varphi_1}(\rho_t),\lim_{t\to\infty}\Rc_{\varphi_2}(\rho_t)\Big\}\\
    &= \max\{\Rc_{\varphi_1}(\rho),\Rc_{\varphi_2}(\rho)\}\\
    &= \max\{\semantics{\rho, \varphi_1}, \semantics{\rho,\varphi_2}\}\\
    &= \semantics{\rho, \varphi_1\lor\varphi_2}.
\end{align*}

\subsubsection{Property $I_2$.} We focus on the case when we have an SCC $\C$ of $\Rc_{\vp}$ such that all states of $\C$ are contained in $Q_1\times Q_2\times \R$ since the other cases are handled in the main paper. Let $Q_i = \bigcup_{\ell=1}^{L_i} Q_\ell^i$ for $i\in\{1,2\}$ be the SCC decompositions of $Q_1$ and $Q_2$ that satisfy property $I_2$. Firstly, there is an $\ell\in[L_1]$ and an $m\in[L_2]$ such that all states of $\C$ are contained in $Q_{\ell}^1\times Q_{m}^2\times \R$. There are two cases to consider.
\begin{itemize}
    \item If $\chi_1(\ell) = \chi_2(m)$, then for any transition $(q_1,q_2,\zeta)\xrightarrow{\sigma}(q_1',q_2',\zeta')$ within $\C$ with $\zeta\neq 0$ we have $|\zeta'| = \frac{|\zeta|}{\lambda} < |\zeta|$. If there is a state $(q_1,q_2,\zeta)\in\C$ with $\zeta\neq 0$ then $\C$ is a singleton with no transitions. Otherwise, for all $(q_1,q_2,\zeta)\in\C$, we have $\zeta=0$ in which case we can conclude that all transitions within $\C$ have reward $\chi_1(\ell) = \chi_2(m)$.\\
    \item If $\chi_1(\ell) \neq \chi_2(\ell)$, let us assume that $\chi_1(\ell) = 1-\lambda$ and $\chi_2(m) = 0$. For any transition $(q_1,q_2,\zeta)\xrightarrow{\sigma}(q_1', q_2',\zeta')$ within $\C$, since $\zeta>-1$, we have
    \begin{align*}
        \zeta' &= \frac{(1-\lambda) + \zeta}{\lambda}\\
        &= \frac{(1-\lambda) + (1-\lambda)\zeta + \lambda\zeta}{\lambda}\\
        &> \frac{(1-\lambda) - (1-\lambda) + \lambda\zeta}{\lambda} > \zeta
    \end{align*}
    which implies that $\C$ cannot have any cycles. Hence $\C$ is a singleton and there are no transitions within $\C$.
\end{itemize}

\subsection{Construction for Eventually Operator}\label{sec:ev_const}
An example of a reward machine constructed using our algorithm is shown in Figure~\ref{fig:eventually}. Recall that $\vp=\F\vp_1$. Note that, in any state $(v,S)$ of $\Rc_{\vp}$, if we have $(q_1,\zeta_1),(q_1,\zeta_1')\in S$ with $\zeta_1<\zeta_1'$ we can remove $(q_1,\zeta_1)$ from $S$ since it will never lead to the maximum value. For clarity of presentation, we assume that this optimization is not performed.
\begin{figure}[t]
    \centering
    \begin{tikzpicture}
        \node[state, initial left] (a0) {$a_0$};
        \node[state] (a1) [above right=1.5cm and 2cm of a0] {$a_1$};
        \node[state] (a2) [below right=1.5cm and 2cm of a0] {$a_2$};
        \node[state] (a3) [right=4cm of a1] {$a_3$};
        \node[state] (a4) [right=4cm of a2] {$a_4$};
        \node[state] (a5) [below left=1.5cm and 2cm of a3] {$a_5$};
        \path[->]
        (a0) edge [] node [] {$p, 1/3$} (a1)
        (a0) edge [swap] node [] {$\neg p, 0$} (a2)
        (a2) edge [loop below] node [] {$\neg p, 0$} ()
        (a2) edge [swap] node [] {$p, 1/3$} (a4)
        (a4) edge [swap] node [] {$p, 1/3$} (a3)
        (a4) edge [] node [xshift=1mm] {$\neg p, 1/3$} (a5)
        (a5) edge [loop, out=195, in=255, looseness=6] node [xshift=-1mm, anchor=east] {$\top, 1/3$} (a5)
        (a1) edge [] node [] {$p, 1/3$} (a3)
        (a1) edge [] node [xshift=-2mm] {$\neg p, 1/3$} (a5)
        (a3) edge [] node [xshift=-2mm] {$\top, 1/3$} (a5)
        ;
    \end{tikzpicture}
    \caption{
    The reward machine for $\vp = \F_\lambda p$ via the eventually construction with $\lambda = 2/3$. See Figure~\ref{fig:atomic_rm} for the reward machine for $p$ used in the construction. The transitions are labeled by the guard and reward. The states are 
    $a_0 = (0, \{(q_0, 0)\})$, 
    $a_1 = (-1/2, \{(q_1, 0), (q_0, -1/2)\})$, 
    $a_2 = (0, \{(q_2, 0), (q_0, 0)\})$, 
    $a_3 = (-1, \{(q_1, 0), (q_1, -3/4)\})$, 
    $a_4 = (-1/2, \{(q_1, 0), (q_2, -1/2), (q_0, -1/2)\})$, and
    $a_5 = (-1, \{(q_2, 0)\})$.
    \label{fig:eventually}
    }
\end{figure}
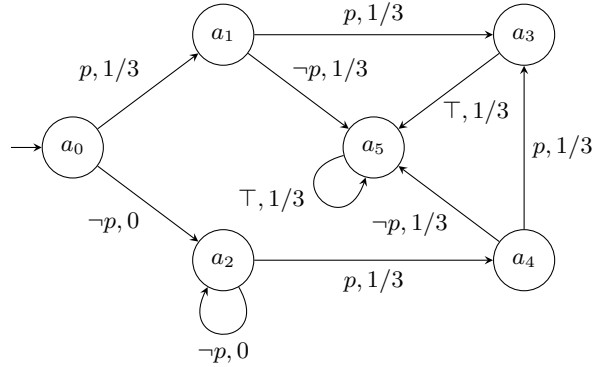

\subsubsection{Finiteness.} From the argument in the main paper, it is sufficient to show that for any state of the form $q = (-1, S_0)\in Q$ the number of states that can be reached from $q$ is finite. Since for any such reachable state $(v,S)$ we have $v=-1$, we simply use $S$ to represent such a state. Let $Q_1=\bigcup_{\ell=1}^L Q_\ell$ be the SCC decomposition of $\Rc_{\vp_1}$ and for $q_1\in Q_\ell$ let $h(q_1)=\ell$ denote the SCC that contains $q_1$. Note that we abuse notation to use $Q_1$ to also denote an SCC of $\Rc_{\vp_1}$ (when clear from context). For any state of $\Rc_{\vp}$ given by $S = \{(q_1,\zeta_1),\ldots, (q_m,\zeta_m)\}$ we define $\mu(S) = \sum_{i=1}^m h(q_i)$. We show that the set of all states reachable from $S_0$ of the form $S$ with $\mu(S) = n$ and $|S| = m$ is finite if the set of all states reachable from $S_0$ of the form $S'$ with either $|S'| > m$ or $|S'|=m$ and $\mu(S') < n$ is finite. This provides an inductive argument for finiteness since in any transition from $S$ to $S'$ either $|S'| < |S|$ or $|S'|=|S|$ and $\mu(S')\geq \mu(S)$. The base case where $|S| = |S_0|$ and $\mu(S) = \mu(S_0)$ also follows from a similar argument. Formally, let $Z_m^n = \{S\mid |S| = m, \mu(S) = n\ \text{and}\ S\ \text{is reachable from}\ S_0\}$.

\begin{lemma}
For any starting state $S_0$ reachable from the initial state $q_0$, the corresponding reach set $Z_m^n$ is finite if $Z_{m'}^{n'}$ is finite for all $m'$ and $n'$ with either (i) $m'>m$ or (ii) $m'=m$ and $n' < n$.
\end{lemma}
\begin{proof}
    We first note that, from the definition of $\delta$, that any reachable state $S$ must contain an element $(q_i,\zeta_i)$ with $\zeta_i=0$. Therefore, if $m=1$ then any $S\in Z_m^n$ equals $\{(q_1,0)\}$ for some state $q_1$ of $\Rc_{\vp_1}$. Hence $Z_m^n$ is finite.
    
    Suppose $m > 1$. To reach a state $S\in Z_m^n$ from $S_0$, we must reach a state $S'\in Z_m^n$ using a single transition from some state $S''\in Z_{m'}^{n'}$ with either (i) $m'>m$ or (ii) $m'=m$ and $n' < n$. Therefore the number of such states $S'$ is finite. Now, it suffices to show that for any such $S'$, the set of states in $Z_m^n$ reachable from $S'$ is finite.
    
    In any path $S' = S_1\to S_2\to\cdots \to S_k$ where $S_k\in Z_m^n$ we have $S_j\in Z_m^n$ for all $1\leq j\leq k$. Furthermore, for each $j$ we have $S_j = \{(q_1^j,\zeta_1^j),\ldots,(q_m^j,\zeta_m^j)\}$ and there is a word $w=\sigma_1\sigma_2\ldots\sigma_{k-1}$ such that $q_i^{j+1} = \delta_1(q_i^{j}, \sigma_j)$ for all $j<k$ and $1\leq i\leq m$. Since we have $\sum_{i=1}^m h(q_i^j) = n$ for all $j$, there must exist indices $\ell_1,\ldots,\ell_m$ such that $q_i^j\in Q_{\ell_i}$ for all $i$ and $j$. There are two cases to consider.
    \begin{itemize}
        \item Suppose the type $c_i$ of $Q_{\ell_i}$ is the same for all $i$. If $\zeta_i^1 = 0$ for all $i$, then $\zeta_i^j = 0$ for all $i$ and $j$ and hence the number of states reachable from $S'$ in $Z_m^n$ is finite. Otherwise, there is an $i$ such that $\zeta_i^1 < 0$. In this case, it is easy to see that $\zeta_i^j = \frac{\zeta_i^1}{\lambda^{j-1}}$. Since $\zeta_i^j \leq -1$ for large enough $j$, the length of the path $k$ is bounded. This is because the element $(q_i^j,\zeta_i^j)$ will be dropped from $S_j$ when $j$ is large enough, thereby reducing the size of $S_j$ and causing the path to leave $Z_m^n$. Therefore, the number of states reachable from $S'$ in $Z_m^n$ is finite.\\
        \item Suppose there exist $i,i'$ with type $c_i = 0$ and $c_{i'}=1-\lambda$. Then we have that $\zeta_{i}^{j+1}{-}\zeta_{i'}^{j+1} = (\zeta_{i}^j{-}\zeta_{i'}^j - (1{-}\lambda))/\lambda$ for all $j{<}k$. Therefore for all $j>1$ we have
        \begin{align*}
            \zeta_{i}^j-\zeta_{i'}^j &= \frac{\zeta_{i}^1-\zeta_{i'}^1}{\lambda^{j-1}} - (1-\lambda)\sum_{x=1}^{j-1}\frac{1}{\lambda^x}\\
            &= \frac{1}{\lambda^{j-1}}\Big(\zeta_{i}^1-\zeta_{i'}^1 - (1-\lambda)\sum_{x=0}^{j-2}\lambda^x\Big).
        \end{align*}
        Since $\zeta_{i}^1-\zeta_{i'}^1 < 1$ we have $\zeta_{i}^j-\zeta_{i'}^j \to -\infty$ as $j\to\infty$. Hence $\zeta_i^j$ must drop below $-1$ for large enough $j$ causing a transition out of $Z_m^n$. Therefore, the length of the path $k$ is bounded which implies that the number of states reachable from $S'$ in $Z_m^n$ is finite.
    \end{itemize}
    This concludes the proof.\hfill\qed
\end{proof}

\subsubsection{Property $I_1$.} The following lemma is sufficient to show this property.
\begin{lemma}
    \label{lem:f_const}
    For any finite word $w \in \Sigma^*$ of length $n$, then if $\delta(q_0, w) = (v, S)$, 
    (i) there is a $(q_k, \zeta_k) \in S$, such that $q_k = \delta_1(q_0^1, w_{i:\infty})$ and 
    \[
    \zeta_k = (\semantics{\X_\lambda^i \vp_1, w} - \semantics{\F_\lambda^n \vp_1, w})/\lambda^n
    \]
    if and only if there is a subformula $\X_\lambda^i \vp_1$ with $n \ge i \ge 0$ and \[
    \semantics{\X_\lambda^i \vp_1, w} + \lambda^n > \semantics{\F_\lambda^n \vp_1, w};
    \]
    
    (ii) $v = \max\{-1, -\semantics{\F_\lambda^n \vp_1, w}/\lambda^n\}$; and

    (iii) $m(S_p, \sigma) = (\semantics{\F_\lambda^{n} \vp_1, u\sigma} - \semantics{\F_\lambda^{n-1} \vp_1, u})/\lambda^{n-1}$ if $w = u\sigma$ where $\sigma \in \Sigma$ and $\delta(q_0, u) = (v_p, S_p)$.
\end{lemma}

\begin{proof}
    By induction on $n$.
    \begin{itemize}
        \item \textit{Base case $(n=0)$.} In this case, $(q_0^1, 0) \in S$ and $0 = \max\{-1, -\semantics{\F_\lambda^n \vp_1, w}/\lambda^n\}$, and $0 = (\semantics{\X_\lambda^i \vp_1, w} - \semantics{\F_\lambda^i \vp_1, w})/\lambda^n$.\\
        
        \item \textit{Inductive step $(n>0)$.} Let $w = u\sigma$ where $\sigma \in \Sigma$, and $\delta(q_0, u) = (v_p, S_p)$. Consider a $(q_k, \zeta_k) \in S_p$. By the inductive hypothesis, there exists an $n-1 \ge i \ge 0$ such that, 
        \begin{align*}
            f(\zeta_k, q_k, \sigma) &= r_1(q_k, \sigma) + \zeta_k \\
            &= r_1(q_k, \sigma) + (\semantics{\X_\lambda^i \vp_1, u} - \semantics{\F_\lambda^{n-1} \vp_1, u})/\lambda^{n-1} \\
            &= (\semantics{\X_\lambda^i \vp_1, u\sigma} - \semantics{\F_\lambda^{n-1} \vp_1, u})/\lambda^{n-1}
        \end{align*}
        where the second equality follows from the inductive hypothesis, and the third equality follows from the fact that $q_k = \delta_1(q_0^1, w_{i:\infty})$ is simulating the reward machine for $\X_\lambda^i \vp_1$ by delaying its start. Now, we have that 
        \begin{align*}
            m(S_p, \sigma) &= \max_{(q_k, \zeta_k) \in S} f(q_k, \zeta_k, \sigma) \\
            &= \max_{i} (\semantics{\X_\lambda^i \vp_1, u\sigma} - \semantics{\F_\lambda^{n-1} \vp_1, u})/\lambda^{n-1} \\
            &= (\semantics{\F_\lambda^n \vp_1, u\sigma} - \semantics{\F_\lambda^{n-1} \vp_1, u})/\lambda^{n-1} .
        \end{align*}
        This shows $(iii)$. 
        
        Next, let $v' = (v_p - m(S_p, \sigma))/\lambda$. If $v_p = -1 = \max\{-1, -\semantics{\F_\lambda^{n-1} \vp_1, u}/\lambda^{n-1}\}$ then $v' \le -1$ and $v = -1 = \max\{-1, -\semantics{\F_\lambda^n \vp_1, w}/\lambda^n\}$. Otherwise, consider $v_p > -1$. We have that
        \begin{align*}
            v' &= (v_p - m(S_p, \sigma))/\lambda \\
            &= ( -\semantics{\F_\lambda^{n-1} \vp_1, u}/\lambda^{n-1} - (\semantics{\F_\lambda^n \vp_1, u\sigma} - \semantics{\F_\lambda^{n-1} \vp_1, u})/\lambda^{n-1})/\lambda \\
            &= - \semantics{\F_\lambda^n \vp_1, u\sigma}/\lambda^n \\
            &= - \semantics{\F_\lambda^n \vp_1, w}/\lambda^n .
        \end{align*}
        Since $v = -1$ if $v' < -1$ by the transition function, we get that $v = \max\{-1, - \semantics{\F_\lambda^n \vp_1, w}/\lambda^n \}$, which shows $(ii)$.
        
        Finally, consider a $(q_k, \zeta_k) \in S_p$ and corresponding $n-1 \ge i \ge 0$ and $\X_\lambda^i \vp_1$. We have that $\Delta(\{(q_k, \zeta_k)\}, \sigma) = A$ where $A = \emptyset$ if $\zeta_k' \le -1$ and $\{(q_k', \zeta_k')\}$ otherwise where 
        \begin{align*}
                \zeta_k' &= (f(\zeta_k, q_k, \sigma) - m(S_p,\sigma))/\lambda \\
                &= (\semantics{\X_\lambda^i \vp_1, u\sigma} - \semantics{\F_\lambda^n \vp_1, u\sigma})/\lambda^n .
        \end{align*}
        If $\zeta_k' \le -1$ then $\semantics{\X_\lambda^i \vp_1, u\sigma} - \semantics{\F_\lambda^n \vp_1, u\sigma} \le -\lambda^n$, which due to property $I_3$ of $\Rc_{\vp_1}$, ensures that $\semantics{\X_\lambda^i \vp_1, u\rho} \le \semantics{\F_\lambda^n \vp_1, u\rho}$ for any continuation $\rho \in \Sigma^\omega$. To finish showing $(i)$, we simply need to note that if $\semantics{\F_\lambda^{n} \vp_1, w} \le \lambda^{n}$, then $v' > -1$, so the element $(q_0^1, v') = (\delta_1(q_0^1, w_{n:\infty}), -\semantics{\F_\lambda^{n} \vp_1, w}/\lambda^n) = (q_0^1, (\semantics{\X_\lambda^n \vp_1, w}-\semantics{\F_\lambda^{n} \vp_1, w})/\lambda^n) $ is added.\hfill\qed
    \end{itemize}
    
\end{proof}

\subsubsection{Property $I_2$.}
From the argument in the main paper, it suffices to show property $I_3$ for the case where $v = v' = -1$. 
Consider an SCC $\C$ in $\Rc_\vp$ with some state $(-1, S) \in \C$. This means that $\delta((-1,S), u) = (-1,S)$ for some $u \in \Sigma^*$ of length $\ell > 0$. Consider a $(q_k, \zeta_k) \in S$. We must have that $q_k$ is in an SCC $\C_k$ of $\vp_1$, otherwise we would not be able to reach $S$ again such that it contains $q_k$. By property $I_2$ of $\Rc_{\vp_1}$, we say that $\C_k$ has a reward type $t_k \in \{0, 1-\lambda\}$.

We now consider the behavior of $\zeta_k$.
By Lemma~\ref{lem:f_const}, we have that $(q_k, \zeta_k) \in S$ corresponds to a subformula $\X_\lambda^i \vp_1$ such that $\zeta_k = (\semantics{\X_\lambda^i \vp_1, w} - \semantics{\F_\lambda^n \vp_1, w})/\lambda^n$ where $w \in \Sigma^*$ is a finite word of length $n$. We have that $\zeta_k = (\semantics{\X_\lambda^i \vp_1, w} - \semantics{\F_\lambda^n \vp_1, w})/\lambda^n = (\semantics{\X_\lambda^i \vp_1, wu} - \semantics{\F_\lambda^{n+\ell} \vp_1, wu})/\lambda^{n+\ell}$, i.e. going all the way around the SCC does not change the value of $\zeta_k$. Say that the type $t_k'$ corresponding to $\semantics{\F_\lambda^{n+\ell} \vp_1, wu}$ differs from $t_k$. Then we have that 
\[
(\semantics{\X_\lambda^i \vp_1, w} - \semantics{\F_\lambda^n \vp_1, w})/\lambda^n \neq (\semantics{\X_\lambda^i \vp_1, wu} - \semantics{\F_\lambda^{n+\ell} \vp_1, wu})/\lambda^{n+\ell}.
\]
It must be that $t_k' = t_k$ and $\zeta_k = 0$. We have that the reward given for this SCC $\C$ is that for $\semantics{\F_\lambda^n \vp_1, w}$, i.e. $t_k'$.

\subsubsection{Property $I_3$.}
By Lemma~\ref{lem:f_const}, we have that for any finite word $w \in \Sigma^*$ of length $n$ and letter $\sigma \in \Sigma$, if $\delta(q_0, w) = (v, S)$, then the reward is $r((v, S), \sigma) = m(S, \sigma) = (\semantics{\F_\lambda^{n+1} \vp_1, w\sigma} - \semantics{\F_\lambda^n \vp_1, w})/\lambda^n$. Let $w' = w \sigma$. We have that
\begin{align*}
    \lambda^n m(S, \sigma) &= \semantics{\F_\lambda^{n+1} \vp_1, w\sigma} - \semantics{\F_\lambda^n \vp_1, w} \\
    &= \max_{n+1\ge k \ge 0} \lambda^k \Rc_{\vp_1}(w'_{k:\infty}) - \max_{n\ge k \ge 0} \lambda^k \Rc_{\vp_1}(w_{k:\infty}) \\
    &\le \max_{q', \sigma'} \lambda^{n} r_1(q', \sigma') \le \lambda^n (1-\lambda)
\end{align*}
where the last inequality follows from the fact that all of the rewards in $\Rc_{\vp_1}$ satisfy property $I_3$ of $\Rc_{\vp_1}$. Since all the rewards in $\Rc_{\vp_1}$ are nonnegative by property $I_3$, we have the $\semantics{\F_\lambda^{n+1} \vp_1, w\sigma} - \semantics{\F_\lambda^n \vp_1, w} \ge 0$. Thus, we have that $0 \le r((v, S), \sigma) = m(S, \sigma) \le 1-\lambda$.

\subsection{Construction for Until Operator}

We omit the detailed proofs for the until operator as they largely follow what was done for the eventually operator, and instead only show the key lemma, which is similar to Lemma~\ref{lem:f_const}.

\begin{lemma}
    \label{lem:u_const}
    For any finite word $w \in \Sigma^*$ of length $n$, then if $\delta(q_0, w) = (v, I, \mathcal{X})$, 

    (i) there is an $S \in \mathcal{X}$ if and only if there is a $\psi_i$ with $n \ge i \ge 0$ such that $\semantics{\psi_i, w} + \lambda^n > \semantics{\vp_1 \U_\lambda^n \vp_2, w}$ where $(q_k, \zeta_k) \in S$ if and only if either:
    \begin{itemize}
        \item there is a subformula $\X_\lambda^i \vp_2$ where $\semantics{\X_\lambda^i \vp_2, w} < \semantics{\vp_1 \U_\lambda^n \vp_2, w} + \lambda^n$, in which case $q_k = \delta_2(q_0^2, w_{i:\infty})$ and $\zeta_k = (\semantics{\X_\lambda^i \vp_2, w} - \semantics{\vp_1 \U_\lambda^n \vp_2, w})/\lambda^n$ or,
        \item there is a subformula $\X_\lambda^j \vp_1$ where $i > j \ge 0$ and $\semantics{\X_\lambda^j \vp_1, w} < \semantics{\vp_1 \U_\lambda^n \vp_2, w} + \lambda^n$, in which case $q_k = \delta_1(q_0^1, w_{j:\infty})$ and $\zeta_k = (\semantics{\X_\lambda^j \vp_1, w} - \semantics{\vp_1 \U_\lambda^n \vp_2, w})/\lambda^n$,
    \end{itemize}

    (ii) $v = \max\{-1, -\semantics{\vp_1 \U_\lambda^n \vp_2, w}/\lambda^n \}$, \\

    (iii) there is an $(q_k, \zeta_k) \in I$ if and only if $\lambda^n > \semantics{\vp_1 \U_\lambda^n \vp_2, w}$ and there is a subformula $\X_\lambda^j \vp_1$ where $n > j \ge 0$ and $\semantics{\X_\lambda^j \vp_1, w} < \semantics{\vp_1 \U_\lambda^n \vp_2, w} + \lambda^n$, in which case $q_k = \delta_1(q_0^1, w_{j:\infty})$ and $\zeta_k = (\semantics{\X_\lambda^j \vp_1, w} - \semantics{\vp_1 \U_\lambda^n \vp_2, w})/\lambda^n$, \\

    (iv) $m(S_p, \sigma) = (\semantics{\vp_1 \U_\lambda^n \vp_2, u\sigma} - \semantics{\vp_1 \U_\lambda^{n-1} \vp_2, u})/\lambda^{n-1}$ if $w = u\sigma$ where $\sigma \in \Sigma$ and $\delta(q_0, u) = (v_p, I_p, \mathcal{X}_p)$.
\end{lemma}

\begin{proof}
    By induction on n.
    \begin{itemize}
        \item \textit{Base case $(n=0)$.} In this case, we have that $\{(q_0^2, 0)\} \in \mathcal{X}$ corresponding to $\vp_2$, $0 = \max\{-1, -\semantics{\vp_1 \U_\lambda^n \vp_2, w}/\lambda^n \}$, there is no $n > j \ge 0$ so $I = \emptyset$, and $(\semantics{\vp_1 \U_\lambda^n \vp_2, u\sigma} - \semantics{\vp_1 \U_\lambda^{n-1} \vp_2, u})/\lambda^{n-1} = 0$. \\
        
        \item \textit{Inductive step $(n > 0)$.} Let $w = u\sigma$ where $\sigma \in \Sigma$ and $\delta(q_0, u) = (v_p, I_p, \mathcal{X}_p)$. By the inductive hypothesis, for each $S_p \in \mathcal{X}_p$ there is an $n-1 \ge i \ge 0$ such that
        \begin{align*}
            f_*(\zeta_{k_2}, q_{k_2}, \sigma) &= r_2(q_{k_2}, \sigma) + \zeta_{k_2} \\
            &= r_2(q_{k_2}, \sigma) + (\semantics{\X_\lambda^i \vp_2, u} - \semantics{\vp_1 \U_\lambda^{n-1} \vp_2, u})/\lambda^{n-1} \\
            &= (\semantics{\X_\lambda^i \vp_2, u\sigma} - \semantics{\vp_1 \U_\lambda^{n-1} \vp_2, u})/\lambda^{n-1}
        \end{align*}
        where the last equality follows from the fact that $q_{k_2} = \delta_2(q_0^2, u_{i:\infty})$ is simulating the reward machine for $\X_\lambda^i \vp_2$, and/or there is an $i > j \ge 0$ such that
        \begin{align*}
            f_*(\zeta_{k_1}, q_{k_1}, \sigma) &= r_1(q_{k_1}, \sigma) + \zeta_{k_1} \\
            &= r_1(q_{k_1}, \sigma) + (\semantics{\X_\lambda^j \vp_1, u} - \semantics{\vp_1 \U_\lambda^{n-1} \vp_2, u})/\lambda^{n-1} \\
            &= (\semantics{\X_\lambda^j \vp_1, u\sigma} - \semantics{\vp_1 \U_\lambda^{n-1} \vp_2, u})/\lambda^{n-1}
        \end{align*}
        where the last equality follows from the fact that $q_{k_1} = \delta_1(q_0^1, u_{j:\infty})$ is simulating the reward machine for $\X_\lambda^j \vp_1$. For this $S_p$, we have that
        \begin{align*}
            n(S_p, \sigma) &= \min_{(q_k, \zeta_k) \in S} f_*(\zeta_i, q_i, \sigma) \\
            &= \min\{
            \!\begin{aligned}[t]
            & \semantics{\X_\lambda^i \vp_2, u\sigma} - \semantics{\vp_1 \U_\lambda^{n-1} \vp_2, u}, \\
            & \min_{i > j \ge 0} \semantics{\X_\lambda^j \vp_1, u\sigma} - \semantics{\vp_1 \U_\lambda^{n-1} \vp_2, u}\}/\lambda^{n-1}
            \end{aligned}
             \\
            &= (\semantics{\psi_i, u\sigma} - \semantics{\vp_1 \U_\lambda^{n-1} \vp_2, u})/\lambda^{n-1}
        \end{align*}
        where the second equality follows from the fact that if a particular $\X_\lambda^i \vp_2$ or $\X_\lambda^j \vp_1$ is not present, then from the inductive hypothesis, it is because it does not affect the minimum, i.e. $\semantics{\X_\lambda^i, \vp_2} \ge \semantics{\vp_1 \U_\lambda^n \vp_2, w} + \lambda^n \ge \semantics{\psi_i, w} + \lambda^n$, or $\semantics{\X_\lambda^j \vp_1} \ge \semantics{\vp_1 \U_\lambda^n \vp_2, w} + \lambda^n \ge \semantics{\psi_i, w} + \lambda^n$. We now have that
        \begin{align*}
            m(\mathcal{X}_p, \sigma) &= \max_{S_p \in \mathcal{X}_p} n(S_p, \sigma) \\
            &= \max_i (\semantics{\psi_i, u\sigma} - \semantics{\vp_1 \U_\lambda^{n-1} \vp_2, u})/\lambda^{n-1} \\
            &= (\semantics{\vp_1 \U_\lambda^n \vp_2, u\sigma} - \semantics{\vp_1 \U_\lambda^{n-1} \vp_2, u})/\lambda^{n-1}.
        \end{align*}
        This shows (iv). 
        
        Next let $v' = (v_p - m(\mathcal{X}_p, \sigma))/\lambda$. If $v_p = -1 = \max\{-1, -\semantics{\vp_1 \U_\lambda^{n-1} \vp_2, u}/\lambda^{n-1}\}$ then $v' \le -1$ and $v = \max\{-1, -\semantics{\vp_1 \U_\lambda^n \vp_2, w}/\lambda^n \}$. Otherwise, consider $v_p > -1$. We have that
        \begin{align*}
            v' &= (v_p - m(\mathcal{X}_p, \sigma))/\lambda \\
            &= (-\semantics{\vp_1 \U_\lambda^{n-1} \vp_2, u}/\lambda^{n-1} - (\semantics{\vp_1 \U_\lambda^n \vp_2, u\sigma} - \semantics{\vp_1 \U_\lambda^{n-1} \vp_2, u})/\lambda^{n-1})/\lambda \\
            &= - \semantics{\vp_1 \U_\lambda^n \vp_2, u\sigma}/\lambda^n \\
            &= - \semantics{\vp_1 \U_\lambda^n \vp_2, w}/\lambda^n .
        \end{align*}
        Since $v = -1$ if $v' < -1$ by the transition function, we get that $v = \max\{-1, - \semantics{\vp_1 \U_\lambda^n \vp_2, w}/\lambda^n \}$ which shows (ii).

        Now, consider a $(q_k, \zeta_k) \in I_p$ and corresponding $n-1 > j \ge 0$ and $\X_\lambda^j \vp_1$. We have that $\Delta(\{(q_k, \zeta_k)\}, \sigma, m(\mathcal{X}_p, \sigma)) = A$ where $A = \emptyset$ if $\zeta_k' \ge 1$ and $\{(q_k', \zeta_k')\}$ otherwise where 
        \begin{align*}
            \zeta_k' &= (f_*(\zeta_k, q_k, \sigma) - m(\mathcal{X}_p, \sigma))/\lambda \\
            &= (\semantics{\X_\lambda^j \vp_1, u\sigma} - \semantics{\vp_1 \U_\lambda^n \vp_2, u\sigma})/\lambda^{n} .
        \end{align*}
        If $\zeta_k' \ge 1$ then $\semantics{\X_\lambda^j \vp_1, w} \ge \semantics{\vp_1 \U_\lambda^n \vp_2, w} + \lambda^{n}$.
        To finish showing (iii), we simply need to note that if $\lambda^n > \semantics{\vp_1 \U_\lambda^n \vp_2, w}$, then $v' > -1$, and the element $\Delta(\{(q_0^1, v'\}) = (\delta_1(q_0^1, w_{n-1:\infty}),(\semantics{\X_\lambda^{n-1} \vp_1, w} - \semantics{\vp_1 \U_\lambda^n \vp_2, w})/\lambda^{n})$ is added.

        Finally, we show (i). Consider an $S_p \in \mathcal{X}_p$ and corresponding $n \ge i \ge 0$ and $\psi_i$. Additionally, consider a $(q_k, \zeta_k) \in S_p$ and corresponding formula $\X_\lambda^i \vp_2$ or $\X_\lambda^j \vp_1$ with $i > j \ge 0$. If the corresponding formula is $\X_\lambda^i \vp_2$, then we have that $\Delta(\{(q_k, \zeta_k)\}, \sigma, m(\mathcal{X}_p, \sigma)) = A$ where $A = \emptyset$ if $\zeta_k' \ge 1$ and $\{(q_k', \zeta_k')\}$ otherwise where
        \begin{align*}
            \zeta_k' &= (f_*(\zeta_k, q_k, \sigma) - m(\mathcal{X}_p, \sigma))/\lambda \\
            &= (\semantics{\X_\lambda^i \vp_2, u\sigma} - \semantics{\vp_1 \U_\lambda^{n} \vp_2, u\sigma})/\lambda^{n}.
        \end{align*}
        If $\zeta_k \ge 1$, then $\semantics{\X_\lambda^i \vp_2, w} \ge \semantics{\vp_1 \U_\lambda^{n} \vp_2, w} + \lambda^n$. If the corresponding formula is $\X_\lambda^j \vp_1$, then we have that $\Delta(\{(q_k, \zeta_k)\}, \sigma, m(\mathcal{X}_p, \sigma)) = A$ where $A = \emptyset$ if $\zeta_k' \ge 1$ and $\{(q_k', \zeta_k')\}$ otherwise where
        \begin{align*}
            \zeta_k' &= (f_*(\zeta_k, q_k, \sigma) - m(\mathcal{X}_p, \sigma))/\lambda \\
            &= (\semantics{\X_\lambda^j \vp_1, u\sigma} - \semantics{\vp_1 \U_\lambda^{n} \vp_2, u\sigma})/\lambda^{n}.
        \end{align*}
        If $\zeta_k \ge 1$, then $\semantics{\X_\lambda^j \vp_1, w} \ge \semantics{\vp_1 \U_\lambda^{n} \vp_2, w} + \lambda^n$. This shows that the evolution of each element in $S_p$ is correct according to $\Delta(S_p, \sigma, m)$. Under $T(\mathcal{X}_p, \sigma, m)$, either $\Delta(S_p, \sigma, m)$ is applied, or $\semantics{\psi_i, w} + \lambda^n > \semantics{\vp_1 \U_\lambda^n \vp_2, w}$, in which case $n(S_p,\sigma) \le -1$, and that $S_p$ is removed. To finish showing (i), we note the following. If $\semantics{\psi_n, w} + \lambda^n > \semantics{\vp_1 \U_\lambda^n \vp_2, w}$, then $v' > -1$ and the set $I' \sqcup (q_0^2, v') = I' \sqcup (\delta_2(q_0^2, w_{n:\infty}), (\semantics{\X_\lambda^n \vp_2, w} - \semantics{\vp_1 \U_\lambda^{n} \vp_2, w})/\lambda^{n})$ is added where $I'$ contains the elements corresponding to $\X_\lambda^j \vp_1$ for $n > j \ge 0$ according to (iii) above.

        \qed
    \end{itemize}
\end{proof}

\end{document}